\documentclass[final]{siamonline250211}
\usepackage{graphicx}
\usepackage{setspace}
\usepackage{amsmath}
\usepackage{amsfonts}
\usepackage{geometry}
\usepackage{hyperref}
\usepackage{natbib}
\usepackage{fancyhdr}
\usepackage{titlesec}
\usepackage{booktabs}
\usepackage{multirow}

\usepackage{amssymb}
\usepackage[overload]{textcase}
\usepackage[T1]{fontenc}
\usepackage{longtable}

\usepackage{caption}
\usepackage{subcaption}

\usepackage{color}
\definecolor{mygreen}{rgb}{0,0.6,0}
\definecolor{myblue}{rgb}{0.3,0.2,0.8}
\definecolor{myred}{rgb}{0.8,0.1,0.1}

\usepackage[american]{babel}
\usepackage{csquotes}
\usepackage{mathrsfs}
\usepackage{float}
\usepackage{enumerate}
\usepackage{enumitem}
\newtheorem{remark}{Remark}

\DeclareMathOperator{\E}{\mathbb{E}}
\newcommand{\R}{\mathbb{R}}
\newcommand{\F}{\widehat F}

\newcommand{\bfo}{\mathbf{1}}
\newcommand{\bfz}{\mathbf{0}}

\title{The Long-Only Minimum Variance Portfolio in a One-Factor Market: Theory and Asymptotics}
\author{
Alec Kercheval\thanks{Department of Mathematics, Florida State University, Tallahassee, FL (\email{akercheval@fsu.edu})} \and 
Ololade Sowunmi\thanks{Department of Mathematics, Florida State University, Tallahassee, FL (\email{osowunmi@fsu.edu})
}}

\headers{Long-Only Minimum Variance Portfolio in a One-Factor Market}{Alec Kercheval and Ololade Sowunmi}

\begin{document}

\maketitle

\begin{abstract}
We study the long-only minimum variance (LOMV) portfolio, an object of interest going back to \citet{markowitz1952}, and the simplest non-trivial instance of quantitative portfolio construction. Despite its importance, aspects of the LOMV remain mysterious. Here we examine the base case of a
general one-factor covariance model, and provide an explicit characterization of the solution.
 As a corollary we resolve an open question of \citet{qi2021} for this setting.
In the high-dimensional regime of many assets,  we prove that the active ratio (proportion of assets with non-zero weight in the LOMV portfolio) converges in almost all cases to a limit expressed in terms of the distribution of asset betas.  In particular, when asset betas are all positive and drawn from a continuous distribution, the active ratio tends to zero in the high dimensional limit, and close to zero when the fraction of negative betas is small.

\end{abstract}

\begin{keywords}
minimum variance portfolio, one factor model, long-only constraint
\end{keywords}

\begin{MSCcodes}
91G10, 90C20
\end{MSCcodes}

\section{Introduction}\label{sec:1}

Interest in equity portfolios optimized to have the lowest possible variance goes back at least to \citet{markowitz1952}. The optimal such portfolio depends on the covariances between pairs of assets, and on the particular constraints of interest.

If there are $p$ assets available for investment, we denote by $w= (w_1,\dots,w_p)^\top \in \mathbb{R}^p$ the $p$-dimensional vector of asset weights defining the portfolio.  The global minimum variance (GMV) portfolio $w^{LS}$ denotes the (unique) long-short portfolio solving the simplest minimum variance problem
 \begin{equation}\label{eq:LSprob}
 	      \begin{split}		
 		 &\min_{w \in \mathbb{R}^p}  w^\top \Sigma w \\
           \text{s.t.} \quad     &w^\top\bfo_p= 1,
 	      \end{split}
\end{equation}
where $\Sigma$ denotes the positive definite covariance matrix of asset returns; $\bfo_p$ denotes the vector of dimension $p$ whose every entry is 1; 
and $w^\top \bfo_p =1$ is the full investment condition setting the sum of the weights equal to 1. 
For the long-short portfolio, some of the weights may be negative.

 Our focus is the long-only minimum variance (LOMV) problem 
 \begin{equation}\label{eq: prob 1}
 	      \begin{split}		
 		 &\min_{w \in \mathbb{R}^p}  w^\top \Sigma w \\
           \text{s.t.} \quad     &w^\top\bfo_p= 1\\
             &w_i \geq 0 \text{ } \text{for all  $i=1,2,...,p$.}
 	      \end{split}
\end{equation}
The long-only constraints $w_i \geq 0$ are often required for real investment portfolios due to the complications and costs of short positions. Since the objective function is strictly convex, and the constraint set
$\{w \in \mathbb{R}^p : w^\top \bfo_p =1, w \geq 0 \}$ is non-empty, compact, and convex, there is always a unique LOMV solution, denoted $w^L$.  

The solution ${w^L}$ of \eqref{eq: prob 1} can be contrasted with the solution $w^{LS}$ of the long-short problem. 
It is well known that the portfolio $w^{LS}$ solving problem \eqref{eq:LSprob}  is given by the
simple formula
\begin{equation}
    w^{LS} = \frac{\Sigma^{-1}\bfo_p}{\bfo_p^\top\Sigma^{-1}\bfo_p}.
\end{equation}
The long-only problem \eqref{eq: prob 1} is much harder.

As we show in this article, the main difficulty in problem \eqref{eq: prob 1} is determining which are the active (positive weight) assets in the optimal portfolio, or, equivalently, which are the assets for which the long-only constraints are binding. 
Let
\begin{equation}
    [p] = \{1,2,\dots,p-1,p\} \text{ and } K = \{i \in [p]: w^L_i >0\}.
\end{equation}
 Here $[p]$ is the set of (indices of) all $p$ available assets, and $K$ identifies the active assets in the long-only optimal  portfolio. Let $k \leq p$ denote the number of elements of $K$. Once we have determined $K$, Proposition \ref{thm:wL} below solves the problem explicitly, by expressing
 the active long-only holdings as equal to the long-short minimum risk portfolio holdings corresponding to the reduced set $K$ of assets.
This is intuitively reasonable
because the long-only constraints are not binding on the positive holdings of $w^L$.

This leaves the problem of determining $K$, the set of active assets of $w^L$, which is normally much smaller than the set of positive-weight assets in the long-short portfolio $w^{LS}$ (e.g. see Figure \ref{fig:long_only_long_short} and Theorems \ref{thm:kp} and \ref{thm:convergence}).  In this article we analyze that problem when the covariance matrix comes from a 1-factor (``single index'') model, in which the covariance matrix takes the form
\begin{equation} \label{eq:1factorcov}
\Sigma = \sigma^2 \boldsymbol{\beta} \boldsymbol{\beta}^\top + \Delta,
\end{equation}
where $\sigma^2 >0$ is the variance of the single factor return,   $\boldsymbol{\beta}$ is a $p$-vector of exposures (loadings) of the asset returns to the common  factor, and $\Delta$ is a diagonal matrix of variances $\delta_i^2 > 0, i \in [p]$, of asset-specific returns uncorrelated with the common factor return.

\begin{figure}[H]
    \centering
   \includegraphics[width=0.8\linewidth]{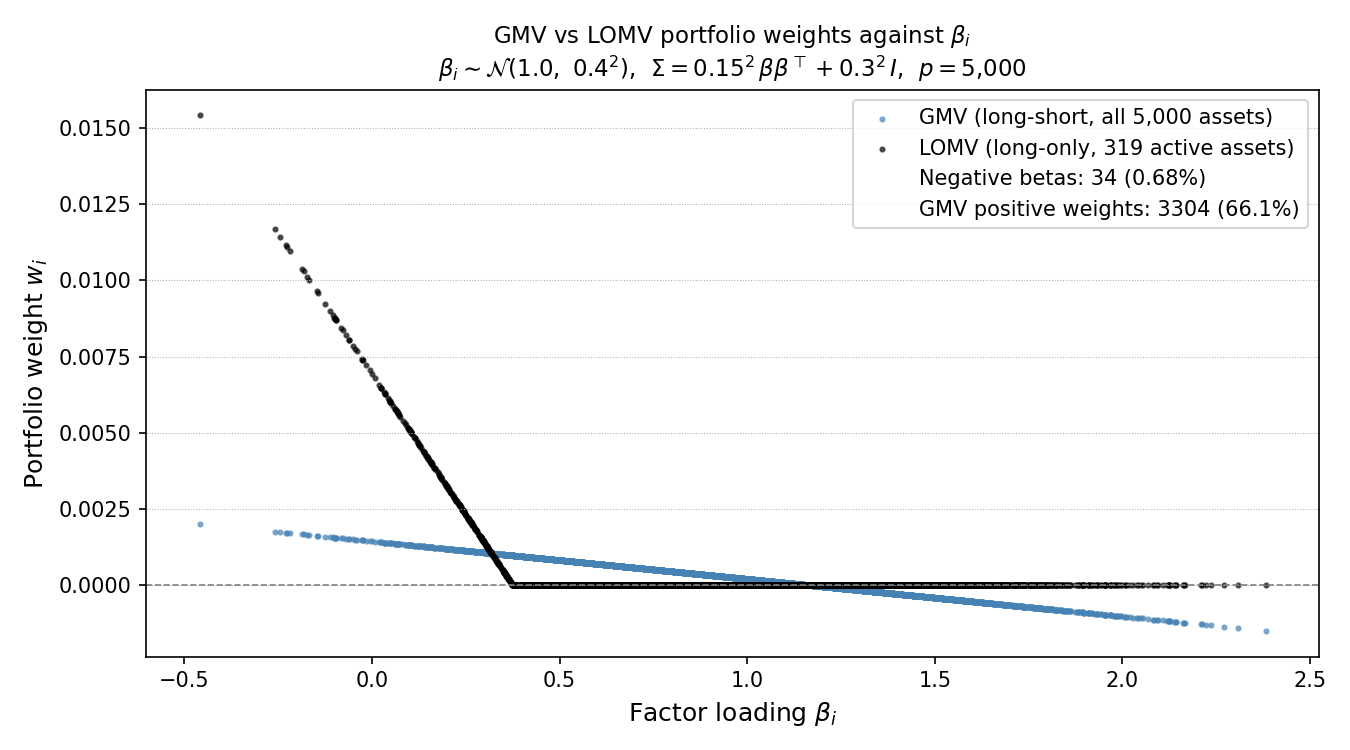}
\caption{Example comparison of long-short (blue) and long-only (black) minimum variance portfolio weights plotted against beta for $p=5000$ assets under the single-factor model \eqref{eq:1factorcov} with $\beta_i$ drawn iid from a Normal $\mathcal{N}(1,\, 0.4^2)$, $\delta_i = 0.3$ for all $i$, and $\sigma = 0.15$ (see Section~\ref{sec:numerics}). The LOMV portfolio contains the 319 active assets with the lowest beta, 34 of which have negative beta. By contrast, 3304 assets (66\%) have positive weight in the GMV portfolio, and all 5000 assets have nonzero weight.}
    \label{fig:long_only_long_short}
\end{figure}

There are several reasons for our focus on the 1-factor model case.  It has been a long-accepted base case of study ever since \citet{sharpe1963}. It remains a workhorse in practice because the dominant eigenvalue  of most equity covariance matrices captures a large fraction of the total variance, and so, as with the capital asset pricing model, it serves as a useful initial approximation to real markets.
The special tractability of the 1-factor model is useful in identifying phenomena present in more complicated models, in which the relationships between portfolio weights, risk, and their dependence on parameters are much less transparent.  Understanding  the complex behavior of more descriptive multi-factor market models must begin with the clearest possible picture in the 1-factor case.  This is the context for and contribution of the present article.

The specific contributions of this article in this general 1-factor setting may be summarized as follows:
\begin{itemize}
    \item We provide a new explicit and easily computable solution (Theorem \ref{thm:1}) for the LOMV problem in terms of the parameters $\sigma^2, \beta_i, \delta_i$, along with a rigorous proof that applies to a general 1-factor model without requiring that the betas are all positive.  This answers an open question stated in \citet{qi2021} concerning the extension to mixed-sign betas. Proposition \ref{thm:wL} describes a general reduction of the problem to identifying the active set $K$, and then Theorem \ref{thm:1} determines 
    the active set by means of an upper threshold on the ordered
betas, identified directly from the sign change of a
computable sequence, without the need to solve a fixed-point
equation.  
    \item We provide a new rigorous analysis of the asymptotic proportion $k_p/p$ of active assets (the {\em active ratio}) in the LOMV portfolio in terms of the asymptotic distribution $F$ of the betas (Theorem \ref{thm:kp}). 
    We show  that in almost all cases,
    \[ \lim_{p \to \infty} k_p/p = F(\beta^*),\]
    where $\beta^* \geq 0$ is the root of an explicit integral equation determined by $F$, but there are rare cases where the limit does not exist. With some moment and concentration bounds on the distribution of betas that are easily satisfied in applications, we also
    show that the asymptotic active ratio
    tends to zero as the fraction of non-positive betas tends to zero
     (Theorem \ref{thm:convergence} and Corollary \ref{cor:activeratio-bound}), via an explicit bound of order $F(0)^{1/3}$ in the fraction $F(0)$ of non-positive betas. 
     This
    supplies
a quantitative explanation for why the long-only constraint commonly
binds on the vast majority of assets in typical large equity markets, since the proportion of negative asset betas is typically small or zero.

\end{itemize}

There is a substantial literature on the LOMV portfolio, e.g. \citep{green-holliffield1992, best1992, jagann2003, cst2011, qi2021}, and references therein.   \citet{green-holliffield1992} show that the dominance of a single factor in asset returns drives the long-short GMV portfolio to take extreme long and short positions.  This suggests that the long-only constraints bind significantly when present.  Our results serve as a rigorous quantitative justification for this conclusion. Since negative market factor loadings are rare empirically, Theorems \ref{thm:kp} and \ref{thm:convergence} explain why the long-only constraint binds on most of the available assets.

\citet{best1992} obtain conditions for when portfolios on the efficient frontier (mean-variance efficient portfolios) can have all positive weights, showing this is possible but the conditions permitting it are fragile, and increasingly so as $p$ grows.  Theorem \ref{thm:1} of the present paper carries this forward for the 1-factor GMV problem by providing an explicit criterion for all positive weights: this will occur when the sequence $\{R_i\}$ of Theorem \ref{thm:1} remains positive.

\citet{jagann2003} show that imposing no-short-sale
constraints in the minimum variance problem with the sample covariance matrix is equivalent, via the KKT
conditions, to solving the long-short problem with a modified
covariance matrix $\tilde{\Sigma}$ whose off-diagonal entries for constrained assets are reduced.  Constructing $\tilde{\Sigma}$ requires knowing the KKT multipliers, which implies knowing the active set $K$, and so this formulation 
does not overcome the need to identify $K$, as we do with Theorem \ref{thm:1}.

Our work was initially inspired by \citet{cst2011}. They study the LOMV in a single-factor market, and provide a semi-explicit solution in terms of a threshold beta defined by a fixed point equation.
\citet{qi2021} points out that their argument is only ``semi-formal'', and relies on the assumption that all the betas are positive.   \citet{qi2021} makes their analysis mathematically rigorous, studying the fixed point problem and proving existence and uniqueness. However, that paper does not remove the restriction $\beta_i >0$ and states as an open problem how to extend the solution to the case when betas can have mixed signs. The present
paper resolves this open question: Theorem~\ref{thm:1} provides a complete,
rigorous solution for the LOMV problem under the 1-factor model with
betas of arbitrary sign.
This characterization of the active set is also instrumental in the proofs of Theorems \ref{thm:kp} and \ref{thm:convergence}.

The problem of $k/p$ asymptotics was introduced to us in the unpublished article \citet{bernstein2025},  which contains a statement similar to our Theorem \ref{thm:kp} but with a different formulation and proof. Their statement employs a stronger finite variance assumption and omits the cases where the limit of $k/p$ fails to exist. That paper additionally contains a finer-grained analysis of the convergence rate of $k/p$ in the case where that limit is zero -- a topic we do not address. Theirs is the first article we know of to examine $k/p$ asymptotics. 
 Our Theorem \ref{thm:convergence} is novel.  A weaker version of Theorem~\ref{thm:1} appeared in the unpublished note \citet{gks2025}, which also contains a discussion of the multi-factor case and issues related to estimation from data that are not addressed here.

Prior work on the relationship between low betas and minimum variance portfolios is not limited to the long-only constraint or one-factor models.
 \citet{steland2024} studies the long-short GMV portfolio in multi-factor
models with uncorrelated factors and nonnegative market betas. He shows that the sign of each
GMV position is determined by comparing $\beta_i$ with an asset-specific
threshold, and confirms empirically on S\&P~500 constituents that the GMV
portfolio tilts toward low-beta, low-idiosyncratic-risk assets. For the
long-only problem his discussion is restricted to the one-factor case, relying on the solution of
\citet{qi2021}, and so retains the restriction $\beta_i \ge 0$ that our
Theorem~\ref{thm:1} removes. 

More generally, for
$\Sigma = B\Omega B^\top + \Delta$ the KKT conditions show that both the
long positions of $w^{LS}$ and the active set of $w^L$ 
consist of the assets whose loading vectors $b_i$ (the rows of $B$) satisfy $b_i^\top f < \lambda$, where $f = \Omega B^\top w$ and $\lambda = w^\top \Sigma w$ for the respective portfolio $w$; in one factor this is the beta threshold
of Corollary~\ref{cor:betathreshold}. In several factors the direction $f$
is not known in advance, so the one-dimensional search of
Theorem~\ref{thm:1} has no direct analogue, and an asymptotic theory of
the active ratio in that setting remains open.

The low-beta composition of the LOMV portfolio
is also relevant to the empirical literature on the ``low-risk anomaly'',
surveyed in \citet{traut2023}, where long-only minimum variance portfolios
\citep{jagann2003, cst2011} are found to deliver market-like returns at
substantially lower risk and to load heavily on low beta and low residual-volatility factors. Our results concern portfolio composition rather than
returns, but within the 1-factor model they give an exact account of that
composition (Corollary~\ref{cor:betathreshold}), and
Theorems~\ref{thm:kp}--\ref{thm:convergence} quantify the portfolio concentration
that motivates the 
diversification constraints often imposed in practice.

The remainder of this article is organized as follows.  Section \ref{sec:main-results} includes more detailed discussion and precise statements of all the theorems. Section \ref{sec:numerics} illustrates our results with some simulation experiments. Proofs of the theorems appear in subsequent sections, with Proposition \ref{thm:wL} and Theorem \ref{thm:1} proved in Section \ref{sec:proof12}, and Theorems \ref{thm:kp} and \ref{thm:convergence} in Section \ref{sec:proof3}. An appendix at the end describes our use of AI tools for this article.

\section{Main Results} \label{sec:main-results}

We state the results outlined above in more detail in this section.  

\subsection{Determining the LOMV portfolio}

We first state a general proposition describing the LOMV portfolio for any positive definite covariance matrix, without regard to whether it comes from a factor model.

\begin{proposition}
\label{thm:wL}
 Suppose the covariance matrix $\Sigma$ of returns for a universe of $p$ assets is an arbitrary $p \times p$ symmetric positive definite matrix. Denote by $w^L$ the solution of problem \eqref{eq: prob 1} and let $K$ denote the set of active assets in $w^L$:
    \begin{equation}
        K = \{i \leq p : w^L_i > 0\},
    \end{equation}
and $k = |K| \leq p$, the number of active assets.
 Let $\Sigma^{K,0}$ be the modified matrix obtained from $\Sigma$ by setting to zero the rows and columns not belonging to $K$.

Then  
\begin{equation}
    w^L  = \frac{(\Sigma^{K,0})^+\bfo_p}{\bfo_p^\top(\Sigma^{K,0})^+\bfo_p}
\end{equation}
where $^+$ denotes the Moore-Penrose inverse.\footnote{For a symmetric matrix $S$ with diagonalization $S = UDU^\top$ for orthogonal $U$ and diagonal $D$, the Moore-Penrose inverse is defined by $S^+ = UD^+U^\top$, where the diagonal matrix $D^+$ is obtained by replacing each nonzero element by its reciprocal, and leaving the zero entries unchanged. When $S$ is positive definite, this coincides with the usual matrix inverse.}

Equivalently, if we denote by  $w^K$ the $k$-vector obtained by deleting all the zero entries of the $p$-vector $w^L$, and by $\Sigma^K$ the $k \times k$ matrix obtained from $\Sigma$ by deleting the rows and columns not belonging to $K$, then
\begin{equation}
    w^K  = \frac{(\Sigma^K)^{-1}\bfo_k}{\bfo_k^\top(\Sigma^K)^{-1}\bfo_k}.
\end{equation}
This is the unique solution of the $k$-dimensional long-short problem \eqref{eq:LSprob} with
$\Sigma$ replaced by $\Sigma^K$, and $w^L$ can be recovered from $w^K$ by adding back zero entries for the deleted assets.
\end{proposition}

Proposition \ref{thm:wL} allows us to determine the long-only minimum risk portfolio as soon as we know the set $K$ of active assets. It remains to determine $K$, and for the single-factor case this can be accomplished by an explicit method described next in Theorem \ref{thm:1}.
\medskip

We henceforth restrict to the case of a single factor covariance matrix $\Sigma$ of the form
\begin{equation} \label{eq:singlefactor}
\Sigma = \sigma^2 \boldsymbol{\beta} \boldsymbol{\beta}^\top + \Delta,
\end{equation}
as discussed above.

We permit the entries $\beta_i$ of $\boldsymbol{\beta}$ to be positive, negative, or zero, but not all zero, which represents the general setting of a 1-factor model of returns.
Since the vector $\boldsymbol{\beta}$ may be replaced by $-\boldsymbol{\beta}$ without changing $\Sigma$, without loss of generality we choose the sign so that 
\begin{equation} \label{eq:positive}
 \sum_{i=1}^p \frac{\beta_i}{\delta_i^2} \geq 0, 
\end{equation}
as would be the case, for example, if all the betas were positive.

In addition, by re-ordering the assets if necessary, for convenience we further assume without loss of generality that the betas are arranged in increasing order:
\begin{equation}
    \beta_1 \leq \beta_2 \leq \cdots \leq \beta_p.
\end{equation}

With these assumptions, our first main result is
\begin{theorem}[Explicit solution to the long-only minimum variance problem under a general 1-factor model] \label{thm:1}
     Suppose the covariance matrix $\Sigma$ has the 1-factor form
     \eqref{eq:singlefactor}.
     
\begin{enumerate}
    \item Let $R_1 = \frac{1}{\sigma^2}$, and, for  $2 \leq i \leq p$, let 
\[
R_i = {\frac{1}{\sigma^2} + \sum_{j=1}^{i-1}\frac{\beta_j}{\delta_j^2}(\beta_j - \beta_i)} = {\frac{1}{\sigma^2} + \sum_{j=1}^{i}\frac{\beta_j}{\delta_j^2}(\beta_j - \beta_i)} .
\]
Then there exists $s \leq p$ such that the initially positive sequence $\{R_i\}$ is monotonically increasing until $i = s$, then monotonically  decreasing for $i > s$. That is,
\begin{equation}
    0 < R_1 \leq R_2 \leq \cdots \leq R_s \geq R_{s+1} \geq \cdots \geq R_p.
\end{equation}
In particular, if
\begin{equation}
    \ell = \max \big\{i: R_i >0\big\},
\end{equation}
then $i \leq \ell$ if and only if $R_i > 0$, and
the
sequence $\{R_i\}$ crosses zero at most once.    

\item Let $w^L$ denote the LOMV solution of problem \eqref{eq: prob 1}, $K = \{i \in [p] : w^L_i >0\}$ denote the set of active assets, and let $k = |K|$ denote the number of active assets.  

Then
\begin{equation} \label{eq:kdef}
   k = \ell \text{ and } K = \{1,2,\dots, \ell\}.
\end{equation}

\end{enumerate}

Let $\Sigma^K$ be the $k \times k$ submatrix of $\Sigma$ consisting of the first $k$ rows and columns,
and let 
\begin{equation} \label{eq:closed_solution}
    w^K = \frac{ (\Sigma^{K})^{-1}\bfo_k}{\bfo_k^\top(\Sigma^{K})^{-1}\bfo_k},
\end{equation}
the long-short fully invested minimum variance solution for the first $k$ assets.

Then, by Proposition \ref{thm:wL}, the solution $w^L = (w_1^L,\dots,w_p^L)^\top$ of  \eqref{eq: prob 1}   is given by
 \begin{equation}
            \begin{split}
                &w^L_i = w^K_i \text{  for $i=1,...,k$}\\
                &w^L_i=0 \text{ for  } i=k+1,\dots,p.\\
            \end{split}
\end{equation}
\end{theorem}

We note that in any case $k>0$, and $k = p$ if the long-short fully invested minimum variance portfolio happens to be long-only already. Otherwise, $k < p$, and $k$ has the property 
\begin{equation}
 R_k > 0 \text{ and }    R_{k+1} = R_k + (\beta_k-\beta_{k+1})\sum_{j=1}^k \frac{\beta_j}{\delta_j^2} \leq 0.
\end{equation}
This means that the threshold index $k$ and the solution $w^L$ are not influenced by the values of $\delta_j$ for $j>k$, nor by the values of $\beta_j$ for $\beta_{j} > \beta_{k+1}$. 
In particular, this implies
\begin{corollary}
    With the notation and assumptions as in the previous theorem, suppose $k < p$ and we add $q>0$ additional assets to the market, resulting in an enlarged set of $p+q$ total assets.  As long as the betas of the added assets are greater than or equal to $\beta_{k+1}$, the active set $K$ and the weights of the active assets in $w^L$ are unaffected.
\end{corollary}

\citet{qi2021} has a related discussion.

The closed form solution \eqref{eq:closed_solution} depends on first determining $k$ from \eqref{eq:kdef}. After sorting the betas, $k$ may be easily found from the monotonicity properties of $\{R_i\}$.

A solution of problem \eqref{eq: prob 1} in semi-explicit form was described for the positive beta case in \citet{cst2011}. They give the following condition:
\begin{equation} \label{eq:thresh}
    w_i > 0 \text{ if and only if } \beta_i < \tau,
\end{equation}
where $\tau$ is the solution of the equation
\begin{equation}
    \tau = \frac{\frac{1}{\sigma^2}+\sum_{\beta_j<\tau} \frac{\beta^2_j}{\delta_j^2}}{\sum_{\beta_j< \tau}\frac{\beta_j}{\delta_j^2}}.
\end{equation}
The active set described this way is a consequence of the following corollary of the proof of Theorem \ref{thm:1}:
\begin{corollary} \label{cor:betathreshold}
    Either $K=[p]$ and the long-short solution $w^{LS}$ is already long-only ($w^{LS} = w^L$), or else $K \neq [p]$, 
    $\sum_{j=1}^k \frac{\beta_j}{\delta_j^2} > 0$, and
    \begin{equation}
        w_i^L > 0 \text{ if and only if } 
        \beta_i < 
        \frac{ \frac{1}{\sigma^2} + \sum_{j=1}^k \frac{\beta_j^2}{\delta_j^2}}{\sum_{j=1}^k \frac{\beta_j}{\delta_j^2}}.
    \end{equation}
\end{corollary}

This formulation requires the convention \eqref{eq:positive}, as a simple argument shows:
 In the absence of any hypotheses about the signs of the betas, replacing $\boldsymbol{\beta}$ by $-\boldsymbol{\beta}$ leaves the covariance matrix $\Sigma$, and hence the solution $w$, unchanged. But in that case the long positions of $w$ would correspond to betas above a threshold, not below it, contradicting the threshold condition \eqref{eq:thresh}.
The need for assumption \eqref{eq:positive} is easily overlooked because negative betas tend to be rare empirically.

\subsection{High dimensional asymptotics of the LOMV portfolio}

Motivated by the large size of typical investment universes,
we wish to study  the proportion $k/p$ of active (non-zero weight) assets  in the LOMV portfolio in a universe of $p$ assets, for large $p$.  A natural approach is to study the asymptotic limit, if it exists.

We fix a single i.i.d.~sequence $(\beta_i,\delta_i^2)_{i\ge1}$ on a probability
space $(\Omega,\mathcal A,\mathbb P)$, with $\beta_i$ independent of
$\delta_i^2$; all almost-sure statements refer to this coupling.
We then consider a sequence of LOMV problems, one for each dimension $p$, where the problem in dimension $p$ uses $(\beta_i,\delta_i^2)_{i\le p}$.

We adopt
the following notation and modeling assumptions:
\begin{itemize}
    \item For each $p$, the i.i.d.~entries $\beta_i$ of the $p$-vector $\boldsymbol{\beta}$ have a distribution with finite mean and cumulative distribution function (cdf) $F$. We denote by $\beta$ a random variable with law $F$ (e.g. $\beta = \beta_1$), and $\E$  denotes expectation under $\mathbb{P}$.
    \item For a given $p$-vector $\boldsymbol{\beta}$, the empirical cdf $\F_p$ of its entries $\beta_i$ is denoted
    \[
    \F_p(t) := (1/p)\, \#\{i \leq p: \beta_i \leq t\}.
    \]
    \item The i.i.d.~diagonal entries $\delta_i^2 = \Delta_{ii} > 0$ have an independent distribution with $\E[1/\delta^2] < \infty$.
    \item The distribution $F$ satisfies:
          \begin{enumerate}
              \item $\E[\beta] \geq 0$ \quad (non-negative mean), and
              \item $\E\!\left[\beta^2\,\mathbf{1}_{\{\beta \le 0\}}\right] < \infty$
                    \quad (finite negative-tail second moment).
          \end{enumerate}
\end{itemize}

The non-negative mean assumption is without loss of generality, since the
covariance matrix is unchanged if $\boldsymbol{\beta}$ is replaced by $-\boldsymbol{\beta}$.
The negative-tail second moment condition is mild: it holds whenever the
left tail of $F$ is sub-Gaussian, sub-exponential, or has any polynomial
decay faster than $|x|^{-2}$.  It is implied by finite variance $\E[\beta^2] <
\infty$ but does not require it, since it also allows distributions with a heavy right tail (such as Pareto with index
$\alpha \in (1,2)$). This allows Theorem~\ref{thm:kp} to cover cases where the cross-sectional betas are heavy-tailed on the right.

We expect the ideas here can also be adapted to the case where the betas and deltas are correlated, which we do not pursue.

Adapting the notation of previous sections,
let $k_p \leq p$ denote the number of active assets in the LOMV solution in dimension $p$, so that the remaining $p - k_p$ weights are zero.  We call the relative number $k_p/p$ of active assets in the LOMV portfolio the {\em active ratio}.  

With the notation and assumptions above, define
\begin{equation}\label{eq:G}
    G(y) \;=\; \int_{-\infty}^{y}\bigl(x^2-yx\bigr)\,dF(x),
    \qquad
    \beta^* \;=\; \sup\bigl\{x\ge 0 : G(x)=0\bigr\}.
\end{equation}

\begin{theorem}[Asymptotic active ratio]\label{thm:kp}
The function $G$ is continuous and concave on $[0,\infty)$. Further, almost surely:

\begin{enumerate}
    \item[Case 1.] \textbf{(Some negative betas, positive mean.)}
          If $\mathbb{P}(\beta<0)>0$ and $\E[\beta]>0$, then $\beta^*$ is the unique zero of $G$, $\beta^* >0$, and
          \begin{equation}\label{eq:exact}
              \liminf_{p\to\infty}k_p/p
              \;=\; F(\beta^{*-}) \;>\; 0,
              \qquad
              \limsup_{p\to\infty}k_p/p
              \;=\; F(\beta^*).
          \end{equation}
          In particular, $\displaystyle\lim_{p\to\infty} k_p/p$ exists
          if and only if $F$ is continuous at~$\beta^*$, in which case
        \begin{equation}
            \lim_{p \to \infty} k_p/p \;=\; F(\beta^*) \;> \;0.
        \end{equation}

    \item[Case 2.] \textbf{(Some negative betas, zero mean.)}
          If $\mathbb{P}(\beta<0)>0$ and $\E[\beta]=0$, then $G(y)>0$ for all
          $y\ge 0$, and
          \[
              \lim_{p\to\infty} k_p/p \;=\; 1.
          \]

    \item[Case 3.] \textbf{(No negative betas.)}
          If $\mathbb{P}(\beta<0)=0$ and $\mathbb{P}(\beta>0)>0$, then
          $\beta^*=\inf\{x\ge 0:F(x)>F(0)\}$ is the left endpoint of
          the support of~$F$ restricted to $(0,\infty)$, and
          \[
              \lim_{p\to\infty} k_p/p \;=\; F(\beta^*)
              \;=\; \mathbb{P}\bigl(\beta\in\{0,\beta^*\}\bigr)
              \quad\text{a.s.}
          \]
          In particular, the limit is zero if and only if $F$ has no
          atom at~$0$ and no atom at~$\beta^*$.  When $F$ has no atom
          at $0$, $\beta^*=\inf\{x\ge0:F(x)>0\}$ and the limit is
          $\mathbb{P}(\beta=\beta^*)$.
\end{enumerate}
\end{theorem}

\begin{remark}\label{rem:convergence}
    In Cases~2 and~3, the limit $\lim k_p/p$ always exists.  In
    Case~1, the limit exists if and only if $F$ is continuous
    at~$\beta^*$; non-convergence requires $\beta^*$ (determined
    implicitly by $G(\beta^*)=0$) to be an atom of~$F$. An example of this is given in Section \ref{subsec:example}, but it is rare. In particular, if $F$ is continuous, as is typical with model distributions, the limit in Case 1 exists, and in Case 3 the limit is zero. 
\end{remark}

Theorem \ref{thm:kp} tells us that the active ratio, when the universe of available assets is large, is determined by the distribution of the asset betas. In particular, a key factor is the proportion $F(0^-) = \mathbb{P}(\beta < 0)$ of negative betas.

In the non-negative beta case $\mathbb{P}(\beta < 0)=0$ treated by \citet{cst2011} and \citet{qi2021}, and in the realistic case where $F$ has no atoms, Theorem \ref{thm:kp} tells us to expect that the active ratio asymptotically vanishes as $p \to \infty$. Even when $\mathbb{P}(\beta <0)$ is positive but small, the result of the next theorem helps to explain the finding in  \citep{best1992} that the long-short minimum variance portfolio is very unlikely to be long-only already. 

In practice one may find that the proportion $\mathbb{P}(\beta \leq 0) = F(0)$ is small but positive, and any negative betas tend to be small in magnitude.  The next theorem shows, under reasonable assumptions, that the limiting active ratio $F(\beta^*)$ is small if the proportion $F(0)$ of non-positive betas is small.
More specifically, $F(\beta^*) \le F(0) + \Theta^{1/3}F(0)^{1/3}$ uniformly over a class of distributions satisfying moment and concentration bounds, with the constant $\Theta$ depending only on those bounds.

\begin{theorem}[Continuity of the limiting active ratio at zero]
\label{thm:convergence}
Fix constants $\mu>0$ and $\kappa,C,M<\infty$. Let
$\mathcal F=\mathcal F(\mu,\kappa,C,M)$ denote the class of probability
distribution functions $F$ on $\R$ such that, with $X\sim F$ and $\E_F$
denoting expectation with respect to $F$:
\begin{enumerate}[label=(\roman*)]
  \item \textbf{positive mean:} $\E_F[X]\ge\mu$;
  \item \textbf{moment bounds:} $\E_F[X^2]\le C$ and
        $\E_F\!\left[X^2\,\mathbf{1}_{\{X\le 0\}}\right]\le \kappa\,F(0)$;
  \item \textbf{concentration bound:} $F(x+t)-F(x)\le Mt$ for all $x \ge0$ and $t>0$.
\end{enumerate}
For $F\in\mathcal F$, let $G$ and $\beta^*$ be defined by \eqref{eq:G}.
Then $\beta^*$ is finite, $\beta^*\le C/\mu$, and
\[
  F(\beta^*)\;\le\;F(0)+\Theta^{1/3}F(0)^{1/3},
  \qquad
  \Theta = 27 \,({\kappa}+ C \sqrt{\kappa}/\mu)\, M^2.
\]
In particular,
$\sup\{F(\beta^*):\,F\in\mathcal F,\ F(0)\le\varepsilon\}\to0$ as
$\varepsilon\downarrow0$.
\end{theorem}

\begin{corollary}\label{cor:activeratio-bound}
With the notation and assumptions of Theorem~\ref{thm:kp}, suppose in
addition that the beta distribution $F$ lies in $\mathcal F(\mu,\kappa,C,M)$.
Then $\lim_{p\to\infty}k_p/p$ exists almost surely, and
\[
  \lim_{p\to\infty}k_p/p \;=\; F(\beta^*)\;\le\;F(0)+\Theta^{1/3}F(0)^{1/3}
  \quad\text{a.s.}
\]
\end{corollary}

\begin{proof}
Since $\E[\beta]\ge\mu>0$, Case~2 of Theorem~\ref{thm:kp} does not occur.
Condition (iii) implies that $F$ has no atoms in $(0,\infty)$. 
In Case~1, $\beta^*>0$, so $F$ is
continuous at $\beta^*$ and $\lim k_p/p=F(\beta^*)$; in Case~3,
$\lim k_p/p=F(\beta^*)$ as well.  Now apply Theorem~\ref{thm:convergence}.
\end{proof}

\begin{remark}
    The moment bounds in (ii) hold for any realistic distribution of betas; for instance, if $F$ is supported in $[-b,b]$ they hold with $C=\kappa=b^2$.  The second bound in (ii) is equivalent to $\E_F[X^2\mid X\le 0]\le \kappa$ when $F(0)>0$, and is vacuous when $F(0)=0$.
\end{remark}

\begin{remark}
The concentration function condition $F(x+t)-F(x)\le Mt$ for $x\ge 0$,
$t>0$ is a mild continuity hypothesis driving the cube-root dependence on $F(0)$. Any
 distribution that is absolutely continuous on $(0,\infty)$ with density bounded by $M$
 satisfies this condition.

The condition imposes nothing on $(-\infty,0]$, so atoms (in particular at 0) and singular parts are allowed there.
  In
particular, it covers distributions $F$ with an atom of mass $m>0$ at zero (representing a population in which a
non-negligible fraction of assets have exactly zero market beta).  If moreover $\mathbb{P}(\beta<0)=0$, so that $F(0)=m$, then Case~3 of Theorem~\ref{thm:kp} gives the exact value
\[
\lim_{p\to\infty} k_p/p = m,
\]
since (iii) rules out an atom at any $\beta^*>0$.  
\end{remark}

\begin{remark}
    We expect that the concentration assumption may be further weakened  to the assumption that the family $\mathcal{F}$ is equicontinuous, with the resulting bound depending on the modulus of continuity.   We do not pursue this generalization.
\end{remark}

The bound of Theorem \ref{thm:convergence} is not sharp, but only serves to establish the continuity claim.  As an example, if the betas across $p \gg 1$ assets are normally distributed with mean 1 and standard deviation 0.4, calculation
gives $\mathbb{P}(\beta <0) \approx 0.62 \%$, or about one asset out of every 160 with a negative beta. We can compute $\beta^* \approx 0.288$ and $F(\beta^*) \approx 3.8\%$ active assets. By comparison, computation of the asymptotic upper bound above gives a value of about 35\% active assets in the LOMV portfolio. 
Unsurprisingly, knowledge of the specific distribution allows for a much sharper estimate of the asymptotic active ratio than our general bound. More comments on this are in Section~\ref{sec:numerics}.

\subsection{Example (non-convergence in Case~1 of Theorem \ref{thm:kp})}
\label{subsec:example}
We construct a simple discrete distribution $F$ for which $\beta^*$ is
an atom, illustrating the non-convergence phenomenon.

Let $F$ be the four-point distribution
\begin{equation}\label{eq:F-example}
    \mathbb{P}(\beta=-1) = 0.05,\quad
    \mathbb{P}(\beta=1) = 0.15,\quad
    \mathbb{P}(\beta=2) = 0.30,\quad
    \mathbb{P}(\beta=5) = 0.50.
\end{equation}
Then $\E[\beta] = -0.05+0.15+0.60+2.50 = 3.20>0$ and
$\mathbb{P}(\beta<0)=0.05>0$, so we are in Case~1.

Computing $G$ on each interval between successive atoms:
\begin{itemize}
    \item $y\in[0,1)$: only the atom at $-1$ contributes.
          $G(y)=0.05(1+y)>0$.
    \item $y\in[1,2)$: atoms at $-1$ and $1$ contribute.
          $G(y) = 0.05(1+y)+0.15(1-y)= 0.20-0.10\,y$,
          which is positive for $y<2$ and equals zero at $y=2$.
    \item $y\in[2,5)$: atoms at $-1$, $1$, and $2$ contribute.
          $G(y) = 0.05(1+y)+0.15(1-y)+0.30\cdot 2(2-y)
               = 1.40-0.70\,y$.
          At $y=2$: $G(2)=0$.  For $y>2$: $G(y)<0$.
    \item $y\ge 5$: all four atoms contribute.
          $G(y) = 1.40-0.70y+0.50\cdot 5(5-y) = 13.90-3.20\,y<0$
          for $y\ge 5$.
\end{itemize}
Hence $G$ has a unique zero at $\beta^*=2$, which is an atom of $F$
with mass $m=0.30$.  The theorem gives, almost surely,
\[
    \liminf_{p\to\infty}\frac{k_p}{p}
    = F(2^-) = 0.05+0.15 = 0.20,
    \qquad
    \limsup_{p\to\infty}\frac{k_p}{p}
    = F(2) = 0.05+0.15+0.30 = 0.50.
\]
  For each finite $p$, the block of
$\approx 0.30\,p$ assets with $\beta_i=2$ is either entirely active
or entirely inactive. The proof of Theorem~\ref{thm:kp} shows that along a single realization, with probability one,
the block switches between active and inactive infinitely often as $p$ tends to infinity.

\section{Numerical Examples} \label{sec:numerics}

In this section we describe the results of numerical simulation studies to illustrate some of the results described above.
We consider the simple covariance model
$\Sigma = \sigma^2 \boldsymbol{\beta} \boldsymbol{\beta}^\top + \delta^2 I$ with $\sigma = 0.15$ and $\delta = 0.3$, parameter values taken from \citet{cst2011}, and with the entries of $\boldsymbol{\beta}$ drawn iid from a given distribution. While we are not undertaking an empirical study, we use these parameters as representative of typical US equity market and specific volatilities, respectively.

Since the LOMV portfolio is unchanged when $\Sigma$ is multiplied by a positive constant, the active set depends on $\sigma$ and $\delta$ only through the ratio $\delta^2/\sigma^2$, which here equals $4$. 

Simulations and figures are generated in Python using code available at 
\begin{center}
\url{https://github.com/akercheval600/LOMV}
\end{center}

\subsection{Figure~\ref{fig:active_ratio_cst}}

For $p = 500, 3000,$ and $10000$, we compute the active ratio $k_p/p$ using Theorem \ref{thm:1} for two distributions of beta: the Normal distribution $\mathcal{N}(1, 0.4^2)$, for which $\mathbb{P}(\beta<0) \approx 0.62\%$, and the same distribution truncated at zero (that is, conditioned on $\beta \ge 0$), so that there are no negative betas. These distributions are chosen to be roughly consistent with the empirical values observed in \citet{steland2024}. There, all betas are observed to be positive, which is modeled by the truncated Normal. We use the full Normal in order to observe the effect of a small number of negative betas.

The first distribution falls under Case~1 of Theorem \ref{thm:kp}, with $\beta^* \approx 0.288$ and $F(\beta^*) \approx 0.0376$; the second falls under Case~3, with $\beta^* = 0$ and $F(\beta^*) = 0$. For each distribution and each $p$ we repeat the experiment 400 times and produce the boxplots shown in Figure \ref{fig:active_ratio_cst}, which compare the distribution of active ratios with the theoretical limit $F(\beta^*)$.

The simulation illustrates
that as the number of assets $p$ increases, the active ratio $k_p/p$ slowly approaches its theoretical limit $F(\beta^*)$.
It also illustrates the role of the negative betas: removing the $0.62\%$ of negative betas by truncation lowers the limit from about $3.8\%$ to zero, as described in Case~3 of Theorem \ref{thm:kp}, and the mean active ratio in the truncated case lies below that of the untruncated case for each $p$.

Figure \ref{fig:active_ratio_cst} also displays a pronounced upward bias phenomenon for moderate values of $p$, not described by the theorems, in which the sample values of $k_p/p$ for a fixed $p$ tend to lie above the limit $F(\beta^*)$. To explain this, as in the proof of Theorem \ref{thm:kp}, for each $p$ we define the function
\[
    G_p(y) = \frac{\nu^2}{p\sigma^2}
             + \frac{\nu^2}{p}\sum_{j=1}^p
               \frac{\beta_j}{\delta_j^2}(\beta_j-y)\,\mathbf{1}_{\{\beta_j\le y\}}
\]
where $1/\nu^2 = \E[1/\delta^2]$ (and in particular $\nu = \delta$ when $\delta_i \equiv \delta$ is constant).
The function $G_p$ is continuous and concave on $(0,\infty)$, and converges to $G$ almost surely as $p \to \infty$, pointwise and in fact uniformly on compact subsets of $(0,\infty)$ (by concavity; see \citet[Theorem~10.8]{rockafellar1970}).  If $\beta^*(p)$ is the unique zero of $G_p$, 
and $\F_p$ is the empirical cdf of the entries of the $p$-vector $\boldsymbol{\beta}$,
we establish in the proof that $k_p/p = \F_p(\beta^*(p)^-)$, $\beta^*(p) \to \beta^*$, and $\F_p \to F$ uniformly on $(0,\infty)$ as $p \to \infty$.

For any fixed $p$,
\[
\E[G_p(\beta^*)] = \frac{\nu^2}{p\sigma^2} + G(\beta^*) = \frac{\nu^2}{p\sigma^2} >0.
\]
Since $G_p(y)$ tends to $-\infty$ as $y \to \infty$, and crosses zero only once, whenever $G_p(\beta^*) > 0$, its zero $\beta^*(p)$ is greater than $\beta^*$, heuristically supporting the observed upward bias of $k_p/p$ for finite $p$.  As $p \to \infty$, the bias term
$\nu^2/(p\sigma^2)$ vanishes, and the sample $k_p/p$ converges as shown.

The heuristic above also suggests a quantitative prediction of $k_p/p$ for finite $p$.
Since $\E[G_p(y)] = \nu^2/(p\sigma^2) + G(y)$ for each fixed $y$, we can replace the random
zero $\beta^*(p)$ of $G_p$ by the zero $y_p$ of its expectation, that is, the solution of
\begin{equation}\label{eq:yp}
    G(y_p) = -\frac{\nu^2}{p\sigma^2}, \qquad y_p > \beta^*
\end{equation}
to obtain $F(y_p)$ as an  estimate for the active ratio at finite $p$.

In Cases~1 and~3 of Theorem~\ref{thm:kp}, $G$ is concave, vanishes at $\beta^*$, and is negative on $(\beta^*,\infty)$. It is therefore strictly decreasing on $[\beta^*,\infty)$, and it tends to $-\infty$. So $y_p$ exists and is unique, and $y_p \downarrow \beta^*$ as $p\to\infty$. This heuristic estimated active ratio $F(y_p)$ depends on $\sigma$ and $\delta$ only through the bias term $\nu^2/(p\sigma^2)$.
Table~\ref{tab:predictor} compares $F(y_p)$ with the mean simulated active ratio from Figure~\ref{fig:active_ratio_cst}. The two agree closely in every case; this is a numerical observation, not a proved statement.

\subsection{Figure~\ref{fig:nonconvergence}}
Figure \ref{fig:nonconvergence} shows the results of 
a simulation experiment specially constructed to illustrate the example in Section \ref{subsec:example} of non-convergence of the active ratio in Case 1 of Theorem \ref{thm:kp}. For each $p \in \{500, 3000, 10000\}$, we simulate 400 trials with betas drawn from the given four-point distribution, chosen so that $F$ has an atom at $\beta^*$. For this experiment we arbitrarily select $\sigma^2 = 1$ and
$\delta^2 = 0.1$. About half the time, $k_p/p$ is near $F(\beta^{*-}) = 0.2$, and the other half near $0.5$.

\begin{figure}[h]
    \centering
    \includegraphics[width=1\linewidth]{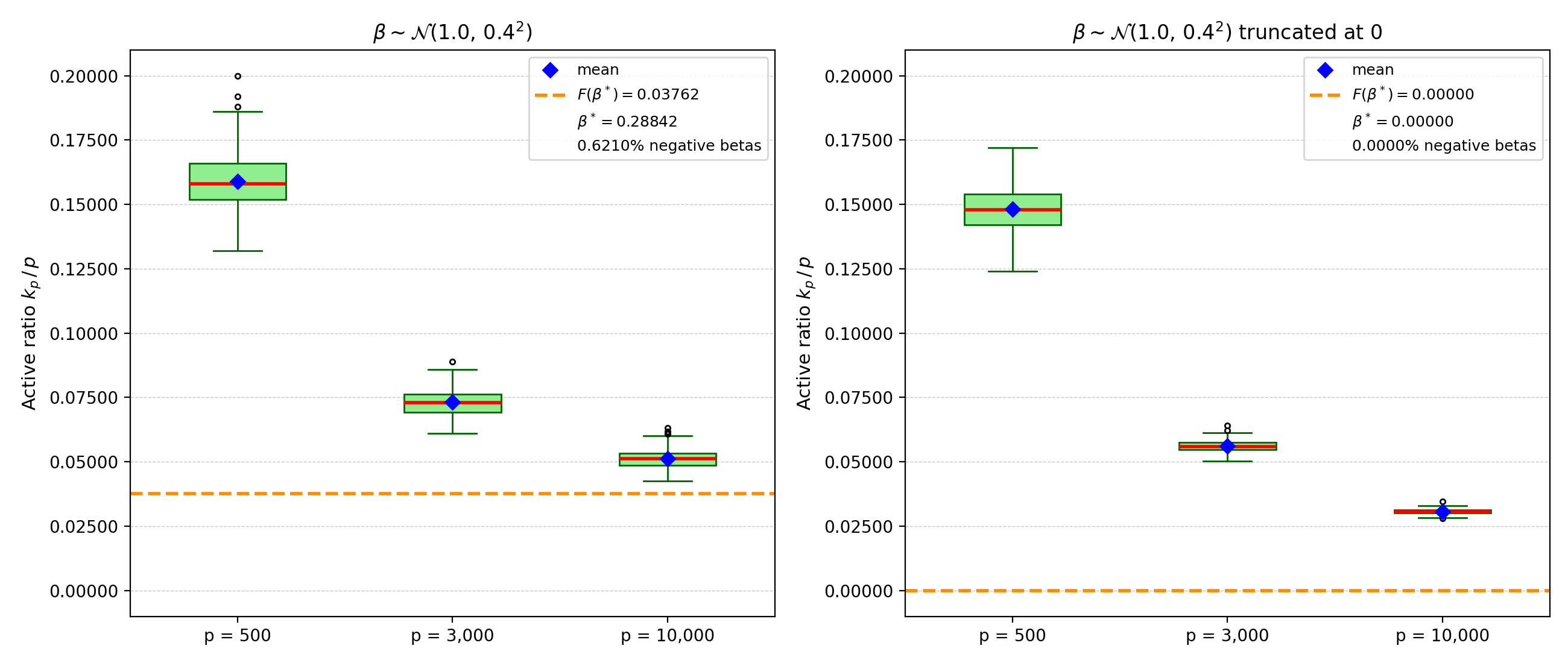}
    \caption{Distribution of the active ratio $k_p/p$ for the LOMV portfolio. \emph{Left:} $\beta_i \sim \mathcal{N}(1.0,\, 0.4^2)$, so that $\mathbb{P}(\beta<0) \approx 0.62\%$. \emph{Right:} $\beta_i$ drawn from $\mathcal{N}(1.0,\,0.4^2)$ conditioned on $\beta_i \ge 0$ (negative draws are redrawn). For each trial, the simulated covariance matrix is $\Sigma = \sigma^2 \boldsymbol{\beta}\boldsymbol{\beta}^\top + \delta^2 I$ with $\sigma = 0.15$ and $\delta = 0.3$, following \citet{cst2011}. For each $p$, 400 simulation trials were performed and the resulting active ratio statistics across trials summarized in boxplots; diamonds mark the means. The dashed line is the theoretical limit $F(\beta^*)$ as $p \to \infty$: Case~1 of Theorem~\ref{thm:kp} on the left, and Case~3, with $\beta^* = 0$ and $F(\beta^*) = 0$, on the right. See text.
}
    \label{fig:active_ratio_cst}
\end{figure}

\begin{table}[ht]
\centering
\caption{Predicted active ratio $F(y_p)$, with $y_p$ defined by \eqref{eq:yp}, compared with the mean simulated active ratio over 400 trials (Figure~\ref{fig:active_ratio_cst}); $\sigma = 0.15$, $\delta = 0.3$, so $\nu^2/(p\sigma^2) = 4/p$.}
\label{tab:predictor}
\begin{tabular}{lc|ccc}
\toprule
Beta distribution & & $p=500$ & $p=3{,}000$ & $p=10{,}000$ \\
\midrule
$\mathcal{N}(1,0.4^2)$ & $y_p$ & 0.6009 & 0.4185 & 0.3468 \\
 & $F(y_p)$ & 0.1592 & 0.0730 & 0.0512 \\
 & simulated mean & 0.1590 & 0.0732 & 0.0512 \\
\midrule
$\mathcal{N}(1,0.4^2)$, truncated & $y_p$ & 0.5891 & 0.3851 & 0.2838 \\
 & $F(y_p)$ & 0.1468 & 0.0563 & 0.0307 \\
 & simulated mean & 0.1481 & 0.0561 & 0.0307 \\
\bottomrule
\end{tabular}
\end{table}

\begin{figure}[htbp]
  \centering
  \includegraphics[width=\textwidth]{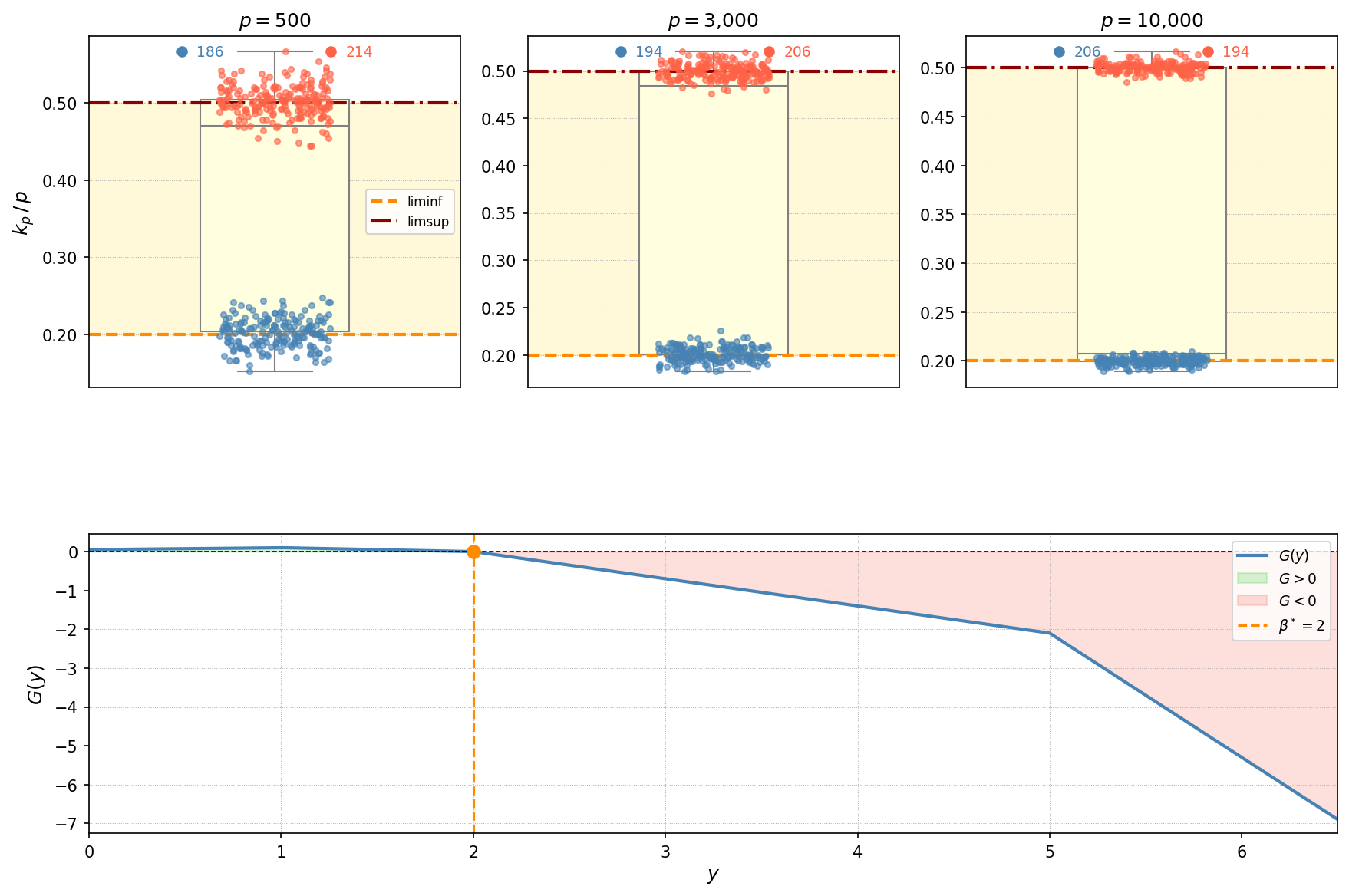}
  \caption{%
    \textbf{Non-convergence example of the LOMV active ratio $k_p/p$ (Theorem \ref{thm:kp}, Case 1).}
    The beta distribution is discrete with four atoms at $-1,1,2,5$ and probabilities
    $\pi_1 = \mathbb{P}(\beta=-1)=0.05$, $\pi_2 =\mathbb{P}(\beta=1)=0.15$, $\pi_3=\mathbb{P}(\beta=2)=0.30$,
    $\pi_4=\mathbb{P}(\beta=5)=0.50$.
     \emph{Bottom panel:} The population $G$ function,
    $G(y)=\sum_i \pi_i\,\beta_i(\beta_i-y)\,\mathbf{1}_{\beta_i\le y}$
    (shaded green where $G>0$ and pink where $G<0$),
    has its unique zero at $\beta^*=2$, which is an atom of $F$ with
    mass $0.30$.\\
    \emph{Top panels:}
    Because $\beta^*$ is an atom, the limit of $k_p/p$ does not exist;
    instead
    $\liminf_{p\to\infty}k_p/p = F(\beta^{*-}) = 0.20$
    and
    $\limsup_{p\to\infty}k_p/p = F(\beta^*) = 0.50$
    (orange dashed and dark-red dash-dot horizontal lines; gold band marks
    the non-convergence interval $[0.20,\,0.50]$).
     Dots are results of $400$ independent Monte Carlo trials for
    each value of $p\in\{500,3000,10\,000\}$, with $\sigma^2 = 1, \delta^2 = 0.1$.
    Blue dots are low-mode trials in which the entire block of
    ${\approx}\,0.30\,p$ assets with $\beta_i=2$ is inactive, so
    $k_p/p\approx 0.20$.
    Red dots are high-mode trials in which that block is wholly active, so
    $k_p/p\approx 0.50$.
  }
  \label{fig:nonconvergence}
\end{figure}

\section{Proofs of Proposition \ref{thm:wL} and Theorem \ref{thm:1}} \label{sec:proof12}

We begin with the proof of the Proposition. The solution $w=w^L$ of our constrained optimization problem \eqref{eq: prob 1} satisfies the well-known Karush-Kuhn-Tucker (KKT) conditions\footnote{See, for example, \citet{boyd-vandenberghe2004} or \citet{beck2023}.}, which are a set of equations in $w$, an auxiliary $p$-vector ${\lambda}$, and an auxiliary scalar $\nu$ (the Lagrange multipliers) as follows:   
    \begin{align}
           &2\Sigma w- {\lambda} +\nu \bfo_p=\bfz_p \label{eq:KKT}\\ 
           & w^\top \bfo_p  =1 \label{eq:KKTb}\\
           &\lambda_i w_i =0 \text{ }i=1,2,...,p \label{eq:KKTc}\\ 
           &\lambda_i \geq 0, w_i \geq 0 \text{ } i=1,2,...,p, \label{eq:KKTd}
    \end{align}
where $\bfz_p$ denotes the $p$-dimensional zero vector.
These KKT conditions include $2p+1$  equations \eqref{eq:KKT} -- \eqref{eq:KKTc} in the $2p+1$ unknowns 
\[
w_1,\dots,w_p,\lambda_1,\dots,\lambda_p, \nu,
\]along
with $2p$ inequality constraints \eqref{eq:KKTd}.

Define the (necessarily non-empty) set
\begin{equation}
    K = \{ i \leq p: w_i > 0 \}.
\end{equation}

Let $k >0$ denote the cardinality of $K$. For any vector $x \in \mathbb{R}^p$, denote by $x^K \in \mathbb{R}^k$ the $k$-dimensional vector obtained from $x$ by deleting the entries $x_j$ for all $j \notin K$.
Likewise, for any $p \times p$ matrix $M$, denote by $M^K$ the $k \times k$ principal submatrix obtained from $M$ by deleting all the rows and columns with indices outside $K$.

Since $\Sigma$ is symmetric positive definite, so is the submatrix $\Sigma^K$, and hence $\Sigma^K$ is invertible. 
Further,
\begin{equation}
    w^\top \Sigma w = (w^K)^\top \Sigma^K w^K.
\end{equation}

Since $\lambda_i = 0$ for all $i \in K$, taking the $k$ rows of equation \eqref{eq:KKT} corresponding to indices in $K$ tells us
\begin{equation} \label{eq:ellKKT}
    2\Sigma^K w^K + \nu \bfo_k = \bfz_k,
\end{equation}
where here the vectors $\bfo_k$ and $\bfz_k$ are $k$-dimensional.  Multiplying \eqref{eq:ellKKT} on the left by $\bfo_k^\top (\Sigma^K)^{-1}$ and using $\bfo_k^\top w^K = 1$, we obtain
\begin{equation} \label{eq:nu}
    \nu = \frac{-2}{\bfo_k^\top (\Sigma^K)^{-1} \bfo_k}
\end{equation}
and therefore
\begin{equation} \label{eq:w*}
    w^K = \frac{(\Sigma^K)^{-1}\bfo_k}{\bfo_k^\top (\Sigma^K)^{-1}\bfo_k} ,
\end{equation}
or, equivalently,
\begin{equation}
    w^L  = \frac{(\Sigma^{K,0})^+\bfo_p}{\bfo_p^\top(\Sigma^{K,0})^+\bfo_p} ,
\end{equation}
completing the proof of Proposition \ref{thm:wL}.

Part 1 of Theorem \ref{thm:1} follows solely from the monotonicity of the sequence $\{\beta_i\}$.
For convenience, define
\begin{equation}
    C_i = \sum_{j=1}^{i}\frac{\beta_j}{\delta_j^2}.
\end{equation} 
It is easy to verify, for all $i= 1,\dots, p-1$, that
\begin{equation}
  R_{i+1} - R_i = (-\beta_{i+1} + \beta_i) C_i.    
\end{equation}
 Since $(-\beta_{i+1} + \beta_i)$ is always non-positive, this means
 \begin{equation}
     C_i \geq 0 \implies R_{i+1} \leq R_i \text{ and }
     C_i \leq 0 \implies R_{i+1} \geq R_i.
 \end{equation}

If $C_i \leq 0$ for all $i \in [p]$, then $\{R_i\}$ is a positive, non-decreasing sequence, so $s=p$ satisfies the conclusion of the theorem.

Otherwise, let $s = \min\{i : C_i > 0\}$, $1 \leq s \leq p$.
If $j < s$, then $C_j \leq 0$ by definition of $s$, so $R_{j+1}\geq R_j$. Since $\beta_s >0$, $C_i$ is increasing for $i \geq s$, so if $j \geq s$ then $C_j > 0$, and we have $R_{j+1} \leq R_j$.  This establishes the conclusion of part 1.
\medskip

We proceed to Part 2. 
Our goal now is to determine $K = \{i \leq p : w^L_i > 0 \}$ in explicit form in terms of the parameters $\sigma, \beta, \delta$ of the problem.

Let $k = |K|$ and define
\begin{equation}
    \ell = \max \{i \leq p: R_i >0 \}.
\end{equation}
We will complete the proof by establishing
\[ k = \ell  \text{ and } K = \{1,2,\dots, k\}.\]

Let $\beta^K \in \mathbb{R}^k$ and the $k \times k$ matrix $\Sigma^K$ be determined from $K$ as before.
For ease of notation write $w = w^L$.

By the Woodbury identity, 
\[
(\Sigma^K)^{-1}= \mathrm{diag}(\frac{1}{(\delta^K)^2})-\frac{(\frac{\beta^K}{(\delta^K)^2})(\frac{\beta^K}{(\delta^K)^{2}})^\top}{\frac{1}{\sigma^2}+ (\frac{\beta^K}{(\delta^K)^2})^\top\beta^K}
\]
where $\frac{1}{({\delta^K})^2} = [\frac{1}{\delta_j^2}: j \in K]^\top$ and
$\frac{\beta^K}{(\delta^K)^{2}}= [\frac{\beta_j}{\delta_j^2}: j\in K]^\top$.

This means
\begin{equation}
    (\Sigma^K)^{-1}\bfo_{k} = \frac{1}{(\delta^K)^2} - \frac{(\frac{\beta^K}{(\delta^K)^2})(\frac{\beta^K}{(\delta^K)^{2}})^\top \bfo_{k}}{\frac{1}{\sigma^2}+ (\frac{\beta^K}{(\delta^K)^2})^\top\beta^K}.
\end{equation}
Now if $i \in K$, then $w^K_i > 0$, so
\begin{equation} \label{eq:vi}
    ((\Sigma^K)^{-1}\bfo_{k})_i = \frac{1}{\delta_i^2} - \frac{\frac{\beta_i}{\delta_i^2}\sum_{j \in K}\frac{\beta_j}{\delta_j^2}}{\frac{1}{\sigma^2}+\sum_{j\in K} \frac{\beta^2_j}{\delta_j^2}} >0.
\end{equation}
Clearing the positive denominators, we
obtain
\begin{equation} \label{eq:beta-condition}
 \frac{1}{\sigma^2}+\sum_{j\in K} \frac{\beta^2_j}{\delta_j^2}- \beta_i\sum_{j\in K}\frac{\beta_j}{\delta_j^2} > 0   
\end{equation}
or
\begin{equation} \label{eq:BK>CK}
    B_K > \beta_i C_K
\end{equation}
where 
\begin{equation}
 B_K = \frac{1}{\sigma^2}+\sum_{j\in K} \frac{\beta^2_j}{\delta_j^2}, \quad C_K = \sum_{j\in K}\frac{\beta_j}{\delta_j^2}.  
\end{equation}

Now suppose instead $i \notin K$, so $w_i=0$. 
Since $\Sigma w = \sigma^2 \beta \beta^\top w + \mathrm{diag}(\delta^2)w$ and $w_j=0$ for all $j \notin K$,
we have
\begin{equation} \label{eq:sigmawi}
    (\Sigma w)_i = \sigma^2 \beta_i \sum_{j \in K} \beta_j w_j.
\end{equation}
Conditions \eqref{eq:KKT} and \eqref{eq:KKTd} tell us
\begin{equation}
    0 \leq \lambda_i = 2(\Sigma w)_i + \nu,
\end{equation}
or, using \eqref{eq:sigmawi},
\begin{equation}
    0 \leq 2 \sigma^2 \beta_i \sum_{j \in K} \beta_j w_j + \nu.
\end{equation}
For $j \in K$, we have, from \eqref{eq:w*},
\begin{equation}
    w_j = \frac{((\Sigma^K)^{-1}\bfo_{k})_j}{\bfo_{k}^\top(\Sigma^K)^{-1}\bfo_{k}}.
\end{equation}
Using this, substituting for $\nu$ with \eqref{eq:nu}, and multiplying through by $\bfo_{k}^\top(\Sigma^K)^{-1}\bfo_{k}/2$ gives us
\begin{equation}
    0 \leq \sigma^2 \beta_i \sum_{j \in K} \beta_j ((\Sigma^K)^{-1}\bfo_{k})_j - 1.
\end{equation}
By \eqref{eq:vi}, 
\begin{equation}
    ((\Sigma^K)^{-1}\bfo_{k})_j = \frac{1}{\delta_j^2} - \frac{\beta_j}{\delta_j^2} \frac{C_K}{B_K}.
\end{equation}
Using this in the previous inequality, we have
\begin{align}
   0 &\leq \sigma^2 \beta_i \sum_{j \in K} \frac{\beta_j}{\delta_j^2} (1 - \frac{\beta_j C_K}{B_K}) -1 \\
   &= \sigma^2 \beta_i \big( C_K - \frac{C_K}{B_K} (B_K - \frac{1}{\sigma^2}) \big) -1 \\
   &= \sigma^2 \beta_i \frac{C_K}{\sigma^2 B_K} -1 = \beta_i \frac{C_K}{B_K} -1,
\end{align}
or
\begin{equation} \label{eq:BK<CK}
    B_K \leq \beta_i C_K.
\end{equation}

Combining \eqref{eq:BK>CK} and \eqref{eq:BK<CK}, we have established the next Lemma:
\begin{lemma} \label{lem:dichotomy}
With $B_K$ and $C_K$ as defined above,
$i \in K$ if and only if $B_K > \beta_i C_K$.
\end{lemma}
With this in hand, we can establish

\begin{lemma}
    $K = \{1,2,\dots,k\}$ and $C_K \geq 0$.
\end{lemma}
\begin{proof}  Suppose that $C_K <0$. Let $\tau = -B_K/C_K >0$. By Lemma \ref{lem:dichotomy}, for all $i \in K$, $j \notin K$, 
\begin{equation}
    \beta_i > -\tau \geq \beta_j.
\end{equation}
Since $C_K = \sum_{i \in K} \frac{\beta_i}{\delta_i^2} <0$, at least one of the terms in this sum must be negative, and hence $\beta_j < 0$ for all $j \notin K$, which implies
\begin{align}
    C_{[p]} = \sum_{j \notin K}\frac{\beta_j}{\delta_j^2} + C_K < 0,
\end{align}
which contradicts our assumption \eqref{eq:positive}. Hence $C_K \geq 0$.

{\em Case 1.} $C_K = 0$. If $j \notin K$, then by Lemma \ref{lem:dichotomy} $B_K \leq \beta_j C_K = 0$, which contradicts $B_K > 0$. Hence $K = [p]$, $k=p$, and $K = \{1,2,\dots,p\}$.

{\em Case 2.} $C_K > 0$. By Lemma \ref{lem:dichotomy}, if $i \in K$, $j \notin K$, then $\beta_i C_K < B_K \leq \beta_j C_K$, and so
\begin{equation}
    \beta_i < \beta_j \text{ for all } i \in K, j \notin K.
\end{equation}
Since $k = |K|$ and the betas are arranged in increasing order, we must have
\begin{equation}
    K = \{1,2,\dots, k\}.
\end{equation}
\end{proof} 

From the previous two lemmas, we immediately have
Corollary \ref{cor:betathreshold}.

It remains to show that ${k} = \ell$. 
It is easy to verify that
\begin{equation}
    R_k = B_K - \beta_k C_K \text{ and } R_{k+1} = B_K - \beta_{k+1} C_K.
\end{equation}
Applying Lemma \ref{lem:dichotomy}, this means $R_k >0$ and $R_{k+1} \leq 0$.  From the monotonicity properties of the sequence $\{R_i\}$ established in Part 1, this implies that
\begin{equation}
    {k} = \max\{i \leq p: R_i > 0 \} = \ell,
\end{equation}
completing the proof of Part 2.
\qed

\section{Proofs of Theorems \ref{thm:kp} and \ref{thm:convergence}} \label{sec:proof3}

In this section we present the proofs of Theorems~\ref{thm:kp} and \ref{thm:convergence} in turn. All limits are taken as $p \to \infty$.

\subsection{Theorem \ref{thm:kp}}

\begin{proof}
We write $1/\nu^2 = \E[1/\delta^2]$ and treat the three cases in turn.
We first establish five preliminary results to be used in each of the three cases. 

\medskip
\noindent\textsl{Preliminary 1: Integral representation, continuity and concavity of $G$.}

Since 
$(x^2-yx)\mathbf{1}_{\{x\le y\}}=-x\,(y-x)^+$ for all real $x,y$,
\[
    G(y)=-\E\bigl[\beta\,(y-\beta)^+\bigr]
        =\E\bigl[(\beta^2-y\beta)\mathbf{1}_{\{\beta\le0\}}\bigr]
        -\E\bigl[\beta\,(y-\beta)^+\mathbf{1}_{\{\beta>0\}}\bigr],
    \qquad y\ge0.
\]
The first expectation is finite by the negative-tail second moment
assumption and finiteness of the mean; the second
satisfies $0\le\beta(y-\beta)^+\mathbf{1}_{\{\beta>0\}}\le y|\beta|$, so is
finite.  In particular, $\E\bigl[|\beta|(y-\beta)^+\bigr]<\infty$ for every
$y\ge0$.

\textit{Integral representation.}  Write
$\psi(t):=\E[\beta\,\mathbf{1}_{\{\beta\le t\}}]$, so that
$|\psi(t)|\le\E|\beta|$ and $\psi$ is non-decreasing on $[0, \infty )$.  For real $x$ and $0\le y<y'$,
\[
    (y'-x)^+-(y-x)^+=\int_y^{y'}\mathbf{1}_{\{x\le t\}}\,dt,
\]
as one checks separately for $x\le y$, $y<x\le y'$, and $x>y'$.  Multiply by
$-\beta$ and take expectations; this is legitimate since $G(y)$ and $G(y')$
are finite.  Since
$\int_y^{y'}\E\bigl[|\beta|\mathbf{1}_{\{\beta\le t\}}\bigr]\,dt\le(y'-y)\,\E|\beta|<\infty$,
Fubini's theorem gives
\begin{equation}\label{eq:Gint}
    G(y')-G(y)=-\int_y^{y'}\psi(t)\,dt,
    \qquad 0\le y<y'.
\end{equation}

\textit{Continuity.}  By \eqref{eq:Gint} and $|\psi|\le\E|\beta|$,
$|G(y')-G(y)|\le\E|\beta|\,|y'-y|$ for $y,y'\ge0$; thus $G$ is Lipschitz,
in particular continuous, on $[0,\infty)$.

\textit{Concavity.}   
Since $\psi$ is non-decreasing, we have, for $0\le y<z<w$ and by \eqref{eq:Gint},
\[
    \frac{G(z)-G(y)}{z-y}
    =-\frac{1}{z-y}\int_y^z\psi(t)\,dt
    \;\ge\;-\psi(z)\;\ge\;
    -\frac{1}{w-z}\int_z^w\psi(t)\,dt
    =\frac{G(w)-G(z)}{w-z}.
\]
A function on an interval whose consecutive chord slopes are
non-increasing in this sense is concave, so $G$ is concave on
$[0,\infty)$.

\medskip
\noindent\textsl{Preliminary 2: Pointwise convergence $G_p(y)\to G(y)$.}
Define the empirical approximation
\[
    G_p(y) = \frac{\nu^2}{p\sigma^2}
             + \frac{\nu^2}{p}\sum_{j=1}^p
               \frac{\beta_j}{\delta_j^2}(\beta_j-y)\,\mathbf{1}_{\{\beta_j\le y\}}.
\]
For each fixed $y\ge0$ the summands
$X_j:=(\beta_j/\delta_j^2)(\beta_j-y)\mathbf{1}_{\{\beta_j\le y\}}
      =-(\beta_j/\delta_j^2)(y-\beta_j)^+$
are i.i.d.  By independence of $\beta$ and $\delta$,
$\E|X_j|=\nu^{-2}\E\bigl[|\beta|(y-\beta)^+\bigr]<\infty$, by the
decomposition in Preliminary~1, and
$\E X_j=-\nu^{-2}\E[\beta(y-\beta)^+]=\nu^{-2}G(y)$.  Since
$\nu^2/(p\sigma^2)\to0$, the SLLN gives $G_p(y)\to G(y)$ almost surely.
The exceptional null set may depend on $y$; Preliminary~5 below
collects countably many such statements into one event of probability one.

\medskip
\noindent\textsl{Preliminary 3: Concavity of $G_p$ on $(0,\infty)$.}

\begin{lemma}\label{lem:Gp}
$G_p$ is continuous and concave on $(0,\infty)$.
\end{lemma}
\begin{proof}
$G_p$ is piecewise affine with breakpoints in $\{\beta_j\}_{j=1}^p$, and
it is continuous on $\R$: the summand indexed by $j$ vanishes at
$y=\beta_j$, so the indicator causes no jump.  For
$y \notin \{\beta_j\}_{j=1}^p$,
$G_p'(y)=-(\nu^2/p)\sum_{j=1}^p(\beta_j/\delta_j^2)\,\mathbf{1}_{\{\beta_j\le y\}}$.
For $y>0$, as $y$ increases the sum $\sum_{j:\beta_j\le y}\beta_j/\delta_j^2$
is non-decreasing (each new term entering as $y$ crosses a positive
$\beta_j$ is positive, and no terms leave).  Hence the slopes of $G_p$ are
non-increasing on $(0,\infty)$, and a continuous piecewise-affine function
with non-increasing slopes is concave.
\end{proof}

The proof of the following lemma is a straightforward exercise.

\begin{lemma}\label{lem:sign}
Let $h\colon[0,\infty)\to\R$ be continuous and concave on $(0,\infty)$,
with $h(0)>0$, and suppose $h$ has a zero in $[0,\infty)$.  Then this
zero $z$ is unique, $z>0$, $h>0$ on $[0,z)$, and $h<0$ on $(z,\infty)$.
\end{lemma}

\medskip
\noindent\textsl{Preliminary 4: Active-set characterization.}

\begin{lemma}\label{lem:activeset}
Suppose $(1/p)\sum_i\beta_i/\delta_i^2>0$ and $G_p$ has a unique
zero $\beta^*(p)\in(0,\infty)$.  Then
\begin{equation}\label{eq:kp-via-betastar}
    k_p = \sum_{i=1}^p\mathbf{1}_{\{\beta_i<\beta^*(p)\}}.
\end{equation}
In particular, $k_p/p = \F_p(\beta^*(p)^-)$, where $\F_p$ is the
empirical cdf.  Moreover, $G_p>0$ on $(-\infty,\beta^*(p))$ and
$G_p<0$ on $(\beta^*(p),\infty)$.
\end{lemma}
\begin{proof}
Since $\sum_i\beta_i/\delta_i^2>0$, Theorem~\ref{thm:1}
applies.  With betas in increasing order,
$k_p=\max\{i\le p:R_i>0\}$ where
$R_i=(1/\sigma^2)+\sum_{j=1}^{i-1}(\beta_j/\delta_j^2)(\beta_j-\beta_i)$.
Therefore
$(\nu^2/p)\,R_i=G_p(\beta_i)$, so $R_i>0$
if and only if $G_p(\beta_i)>0$.  For $y\le0$, every term of the sum in
$G_p(y)$ is non-negative, since $\beta_j\le y\le 0$ implies
$\beta_j(\beta_j-y)\ge0$; hence $G_p(y)\ge\nu^2/(p\sigma^2)>0$.  On
$[0,\infty)$, $G_p$ is continuous, concave on $(0,\infty)$
(Lemma~\ref{lem:Gp}), and $G_p(0)>0$, so Lemma~\ref{lem:sign} gives
$G_p>0$ on $[0,\beta^*(p))$ and $G_p<0$ on $(\beta^*(p),\infty)$.  This
proves the final assertion, and shows that $G_p(\beta_i)>0$ if and only
if $\beta_i<\beta^*(p)$, giving \eqref{eq:kp-via-betastar}.
\end{proof}

\medskip
\noindent\textsl{Preliminary 5: A probability-one event.}
Let $\Omega_0$ be the event on which
\begin{enumerate}[label=($\Omega$\alph*)]
  \item\label{om:G} $G_p(q)\to G(q)$ for every rational $q\ge0$;
  \item\label{om:F} $\sup_{t\in\R}|\F_p(t)-F(t)|\to0$;
  \item\label{om:s} $\frac1p\sum_{i=1}^p\beta_i/\delta_i^2\to\E[\beta]\,\E[1/\delta^2]$.
\end{enumerate}
Item~\ref{om:G} is a countable intersection of probability-one events by
Preliminary~2, \ref{om:F} holds almost surely by the Glivenko--Cantelli
theorem, and \ref{om:s} holds almost surely by the SLLN, since
$\E|\beta/\delta^2|=\E|\beta|\,\E[1/\delta^2]<\infty$ by independence.
Hence $\mathbb{P}(\Omega_0)=1$.  Item~\ref{om:F} implies
$\F_p(t)\to F(t)$ and $\F_p(t^-)\to F(t^-)$ for every $t\in\R$, because
$|\F_p(t^-)-F(t^-)|\le\sup_s|\F_p(s)-F(s)|$.

\emph{Convention.}  In each case below we fix an arbitrary outcome in $\Omega_0$,
intersected with further probability-one events where indicated.  All
subsequent statements are then deterministic, and ``for all large $p$''
means ``for all $p\ge p_0$'', where $p_0$ may depend on the outcome and
on auxiliary parameters such as $\epsilon$ or $M$.  Since the conclusions
hold at every such outcome, they hold almost surely.

\medskip
\noindent\underline{\textsl{Case 1: $\mathbb{P}(\beta<0)>0$ and $\E[\beta]>0$.}}
The proof of Case 1 proceeds in five steps.

\medskip
\noindent\textit{Step~1: Properties of $G$.}\;
$G(0)=\E[\beta^2\mathbf{1}_{\{\beta\le 0\}}]>0$.  Dividing
$G(y)=\E[\beta^2\mathbf{1}_{\{\beta\le y\}}]-y\E[\beta\mathbf{1}_{\{\beta\le y\}}]$
by $y>0$ and letting $y\to\infty$: the second term increases to
$\E[\beta]$ by monotone convergence; the first term $\E[\beta^2\mathbf{1}_{\{\beta\le y\}}]/y$
tends to zero by dominated convergence (for $y\ge1$ the integrand is
bounded by $\beta^2\mathbf{1}_{\{\beta\le0\}}+\beta^+$, which is integrable,
and it vanishes pointwise).  Hence $G(y)/y\to-\E[\beta]<0$,
so $G(y)\to-\infty$.  By Preliminary~1 and the intermediate value
theorem, $G$ has a zero in $(0,\infty)$, and by Lemma~\ref{lem:sign} this
zero $\beta^*$ is unique, with $G>0$ on $[0,\beta^*)$ and $G<0$ on
$(\beta^*,\infty)$.
Since $\beta^*>0$ and $\mathbb{P}(\beta<0)>0$, the entire negative-beta mass
lies below $\beta^*$, giving $F(\beta^{*-})\ge \mathbb{P}(\beta<0)>0$.

\medskip
\noindent\textit{Step~2: Existence, uniqueness, and convergence of $\beta^*(p)$.}

By~\ref{om:s} and $\E[\beta]>0$, $(1/p)\sum_j\beta_j/\delta_j^2>0$ for all
large $p$.  For such $p$,
\[
G_p(0)=\nu^2/(p\sigma^2)+(\nu^2/p)\sum_{j:\beta_j\le 0}\beta_j^2/\delta_j^2>0,
\]
and for $y\ge\max_j\beta_j$,
\[
G_p(y)=\nu^2/(p\sigma^2)+(\nu^2/p)\sum_j\beta_j^2/\delta_j^2
-y\cdot(\nu^2/p)\sum_j\beta_j/\delta_j^2\;\to\;-\infty
\quad\text{as } y\to\infty.
\]
By the intermediate value theorem $G_p$ has a zero in $(0,\infty)$, and
by Lemmas~\ref{lem:Gp} and~\ref{lem:sign} it is unique; we denote it
$\beta^*(p)$.

Now fix $\epsilon\in(0,\beta^*)$ and choose rationals
$q_1\in(\beta^*-\epsilon,\beta^*)$ and $q_2\in(\beta^*,\beta^*+\epsilon)$.
By Step~1, $G(q_1)>0>G(q_2)$, so by~\ref{om:G},
$G_p(q_1)>0>G_p(q_2)$ for all large $p$.  By Lemma~\ref{lem:sign} applied
to $G_p$, this gives $q_1<\beta^*(p)<q_2$, hence
$|\beta^*(p)-\beta^*|<\epsilon$, for all large $p$.  Thus
$\beta^*(p)\to\beta^*$.

\medskip
\noindent\textit{Step~3: Upper bound $\limsup_{p\to\infty} k_p/p \le F(\beta^*)$.}\;
For all large $p$, Step~2 and Lemma~\ref{lem:activeset} give
$k_p/p = \F_p(\beta^*(p)^-)$.
Fix $\epsilon>0$.  Since $\beta^*(p)\to\beta^*$, for all large $p$ we
have $\beta^*(p)<\beta^*+\epsilon$, so
\[
    \frac{k_p}{p}
    = \F_p(\beta^*(p)^-)
    \le \F_p(\beta^*+\epsilon).
\]
By~\ref{om:F}, $\F_p(\beta^*+\epsilon)\to F(\beta^*+\epsilon)$.
Hence $\limsup_{p\to\infty} k_p/p \le F(\beta^*+\epsilon)$ for every
$\epsilon>0$, and right-continuity of $F$ gives
\begin{equation}\label{eq:upper}
    \limsup_{p\to\infty}\frac{k_p}{p}
    \;\le\; F(\beta^*).
\end{equation}

\medskip
\noindent\textit{Step~4: Lower bound $\liminf_{p\to\infty} k_p/p \ge F(\beta^{*-})$.}\;
Fix $\epsilon\in(0,\beta^*)$.  By Step~2, for all large $p$ we have
$\beta^*(p)>\beta^*-\epsilon$, so every $\beta_i\le\beta^*-\epsilon$
satisfies $\beta_i<\beta^*(p)$ and is active by
Lemma~\ref{lem:activeset}.  Hence
\[
    \frac{k_p}{p}
    = \F_p(\beta^*(p)^-)
    \;\ge\; \frac{1}{p}\#\{i:\beta_i\le\beta^*-\epsilon\}
    = \F_p(\beta^*-\epsilon).
\]
By~\ref{om:F}, $\F_p(\beta^*-\epsilon)\to F(\beta^*-\epsilon)$.
Hence $\liminf_{p\to\infty} k_p/p \ge F(\beta^*-\epsilon)$ for every
$\epsilon\in(0,\beta^*)$, and since $F(\beta^*-\epsilon)\uparrow F(\beta^{*-})$
as $\epsilon\downarrow0$,
\begin{equation}\label{eq:lower}
    \liminf_{p\to\infty}\frac{k_p}{p}
    \;\ge\; F(\beta^{*-}).
\end{equation}

\medskip
\noindent\textit{Step~5: Both bounds are achieved.}\;
Set $m = \mathbb{P}(\beta=\beta^*)$.  If $m=0$, then
$F(\beta^{*-})=F(\beta^*)$, and \eqref{eq:upper} and \eqref{eq:lower}
give $\lim_{p\to\infty}k_p/p=F(\beta^*)$; both equalities
in~\eqref{eq:exact} hold.

Now suppose $m>0$.  Let $X_j$ denote the summands of Preliminary~2
at $y=\beta^*$, and let $S_p$ be their partial sums:
\[
    X_j := \frac{\beta_j}{\delta_j^2}\,(\beta_j-\beta^*)\,
           \mathbf{1}_{\{\beta_j\le\beta^*\}},
    \qquad
    S_p := \sum_{j=1}^p X_j,
\]
so that
\begin{equation}\label{eq:Hp}
    H_p \;:=\; G_p(\beta^*)
    \;=\; \frac{\nu^2}{p}\Bigl(\frac{1}{\sigma^2}+S_p\Bigr).
\end{equation}
Since the pairs $(\beta_j,\delta_j)$ are i.i.d., the $X_j$ are i.i.d.\
under $\mathbb{P}$.  
Moreover,
\begin{enumerate}[label=(X\arabic*)]
  \item\label{X:int} $\E|X_1|<\infty$, by Preliminary~2;
  \item\label{X:mean} $\E X_1 = 0$: by independence of $\beta$ and $\delta$,
        \[
        \E X_1 = \E[\delta^{-2}]\,
        \E\bigl[\beta(\beta-\beta^*)\mathbf{1}_{\{\beta\le\beta^*\}}\bigr]
        = \nu^{-2}\,G(\beta^*) = 0;
        \]
  \item\label{X:nondeg} $X_1$ is not almost surely constant: on the event
        $\{\beta_1<0\}$, which has positive probability, both $\beta_1$
        and $\beta_1-\beta^*$ are negative, so $X_1>0$.  By~\ref{X:mean}
        the only possible constant value is $0$, which is therefore excluded.
\end{enumerate}

Thus $(S_p)_{p\ge1}$ is a random walk on $\R$ with integrable, mean-zero,
non-degenerate increments.   
A random walk on $\R$ satisfies, almost surely,
exactly one of: $S_p=0$ for all $p$; $S_p\to+\infty$; $S_p\to-\infty$;
or $\liminf_p S_p=-\infty$ and $\limsup_p S_p=+\infty$
\citep[Exercise~5.4.1]{durrett2019} (following from the standard Hewitt-Savage 0-1 law).  The first alternative would force
$X_1=S_1=0$ almost surely, contradicting~\ref{X:nondeg}. 
Since $S_p/p\to0$ in probability by the law of large numbers,
the walk is recurrent by the Chung--Fuchs theorem
\citep{chung-fuchs1951}; see also \citet[Theorem~5.4.8]{durrett2019}.
Recurrence excludes $S_p\to\pm\infty$.  Hence
\begin{equation}\label{eq:oscillation}
    \liminf_{p\to\infty} S_p = -\infty,
    \qquad
    \limsup_{p\to\infty} S_p = +\infty
    \qquad\text{a.s.}
\end{equation}
By~\eqref{eq:Hp}, it follows that almost surely $H_p>0$ for infinitely
many $p$ and $H_p<0$ for infinitely many $p$.

In accordance with the convention of Preliminary~5, we now fix an arbitrary
outcome in $\Omega_0$ at which~\eqref{eq:oscillation} also holds.  For all
large $p$, Step~2 and Lemma~\ref{lem:activeset} give $G_p>0$ on
$(-\infty,\beta^*(p))$ and $G_p<0$ on $(\beta^*(p),\infty)$.  Write
$\F_p(\beta^{*-})=\frac1p\#\{i:\beta_i<\beta^*\}$; by~\ref{om:F},
$\F_p(\beta^*)\to F(\beta^*)$ and $\F_p(\beta^{*-})\to F(\beta^{*-})$.

\smallskip
\noindent\emph{Indices $p$ with $H_p>0$.}\;
Then $\beta^*<\beta^*(p)$, so every $\beta_i\le\beta^*$ is active and
\[
    \frac{k_p}{p} \;\ge\; \F_p(\beta^*).
\]
These indices form an infinite set along which
$\F_p(\beta^*)\to F(\beta^*)$, so
$\limsup_{p\to\infty}k_p/p\ge F(\beta^*)$.  Combined
with~\eqref{eq:upper}, $\limsup_{p\to\infty} k_p/p = F(\beta^*)$.

\smallskip
\noindent\emph{Indices $p$ with $H_p<0$.}\;
Then $\beta^*(p)<\beta^*$, so every $\beta_i\ge\beta^*$ is inactive and
\[
    \frac{k_p}{p}
    = \F_p(\beta^*(p)^-)
    \;\le\; \F_p(\beta^{*-}).
\]
These indices form an infinite set along which
$\F_p(\beta^{*-})\to F(\beta^{*-})$, so
$\liminf_{p\to\infty}k_p/p\le F(\beta^{*-})$.  Combined
with~\eqref{eq:lower}, $\liminf_{p\to\infty} k_p/p = F(\beta^{*-})$.

\smallskip
Since $F(\beta^*)-F(\beta^{*-})=m>0$, the $\liminf$ and $\limsup$ are
distinct, so the limit does not exist.  This completes Case~1.

\medskip
\noindent\underline{\textsl{Case 2: $\mathbb{P}(\beta<0)>0$ and $\E[\beta]=0$.}}

\medskip
Since $\E[\beta]=0$, for every $t\ge0$ we have
$-\psi(t)=-\E[\beta\,\mathbf{1}_{\{\beta\le t\}}]
       =\E[\beta\,\mathbf{1}_{\{\beta>t\}}]\ge 0$,
so by~\eqref{eq:Gint} $G$ is non-decreasing on $[0,\infty)$.  Since
$G(0)=\E[\beta^2\mathbf{1}_{\{\beta\le 0\}}]>0$, it follows that
$G(y)\ge G(0)>0$ for all $y\ge 0$, so $G$ has no zero on $[0,\infty)$.

Here the sign of $s_p:=\sum_{i=1}^p\beta_i/\delta_i^2$ need not
stabilize, since $s_p/p\to\E[\beta]\,\E[1/\delta^2]=0$, so the
convention~\eqref{eq:positive} may require replacing $\boldsymbol{\beta}$
by $-\boldsymbol{\beta}$ for some values of $p$.  This replacement leaves
$\Sigma$, and hence $k_p$, unchanged.  Let
\[
    \tilde G_p(y) := \frac{\nu^2}{p\sigma^2}
      +\frac{\nu^2}{p}\sum_{j=1}^p\frac{\beta_j}{\delta_j^2}(\beta_j+y)\,
       \mathbf{1}_{\{\beta_j\ge -y\}}
\]
denote the function $G_p$ computed from $-\boldsymbol{\beta}$.  By
Theorem~\ref{thm:1} and the identity $(\nu^2/p)R_i=G_p(\beta_i)$ from the
proof of Lemma~\ref{lem:activeset}: if $s_p\ge0$, asset $i$ is active if
and only if $G_p(\beta_i)>0$; if $s_p<0$, asset $i$ is active if and only
if $\tilde G_p(-\beta_i)>0$.

Fix a positive integer $M$.  Since $G(M)\ge G(0)>0$, \ref{om:G} gives
$G_p(M)>0$ for all large $p$.  For $\tilde G_p$, the summands
$Y_j:=(\beta_j/\delta_j^2)(\beta_j+M)\mathbf{1}_{\{\beta_j\ge-M\}}$ are
i.i.d., and they are negative only when $-M<\beta_j<0$, where
$|Y_j|\le M^2/\delta_j^2$; hence $\E[Y_1^-]<\infty$.  Moreover, by
independence,
\[
    \E[Y_1]=\E[1/\delta^2]\,
      \Bigl(\E\bigl[\beta^2\mathbf{1}_{\{\beta\ge-M\}}\bigr]
      +M\,\E\bigl[\beta\,\mathbf{1}_{\{\beta\ge-M\}}\bigr]\Bigr)\in(0,\infty],
\]
because $\E[\beta\mathbf{1}_{\{\beta\ge-M\}}]=-\E[\beta\mathbf{1}_{\{\beta<-M\}}]\ge0$
(as $\E[\beta]=0$), and
\[
\E[\beta^2\mathbf{1}_{\{\beta\ge-M\}}]\ge\E[\beta^2\mathbf{1}_{\{\beta>0\}}]>0
\]
(as $\E[\beta]=0$ and $\mathbb{P}(\beta<0)>0$ force $\mathbb{P}(\beta>0)>0$).
The positive-tail second moment of $\beta$ is not assumed finite, so
$\E[Y_1]=+\infty$ is possible; the SLLN still applies (via truncation from above), since
$\E[Y_1^-]<\infty$, and gives $\frac1p\sum_{j\le p}Y_j\to\E[Y_1]$ almost
surely.  We now fix an outcome in $\Omega_0$ at which this convergence
holds for every positive integer $M$.  Then for each such $M$,
$\tilde G_p(M)\ge(\nu^2/p)\sum_{j\le p}Y_j>0$ for all large $p$.

Fix $M$ and let $p$ be large enough that $G_p(M)>0$ and $\tilde G_p(M)>0$.
If $s_p\ge0$: for $\beta_i\le0$, $G_p(\beta_i)>0$ (every term of the sum
is non-negative, as in the proof of Lemma~\ref{lem:activeset}); for
$0<\beta_i\le M$, $G_p(\beta_i)>0$ because $G_p$ is continuous and concave
on $[0,M]$ (Lemma~\ref{lem:Gp}) and positive at both endpoints.  So every
asset with $\beta_i\le M$ is active.  If $s_p<0$, the same argument
applied to $\tilde G_p$, which is $G_p$ computed from $-\boldsymbol{\beta}$,
shows that every asset with $\beta_i\ge -M$ is active.  In either case
\[
    \frac{k_p}{p}
    \;\ge\; \min\bigl\{\F_p(M),\,1-\F_p((-M)^-)\bigr\}
    \;\longrightarrow\; \min\bigl\{F(M),\,1-F((-M)^-)\bigr\}
\]
by~\ref{om:F}.  Hence
$\liminf_{p\to\infty}k_p/p\ge\min\{F(M),1-F((-M)^-)\}$ for every positive
integer $M$.  Letting $M\to\infty$, and using $k_p\le p$, we conclude
$k_p/p\to 1$.

\medskip
\noindent\underline{\textsl{Case 3: $\mathbb{P}(\beta<0)=0$ and $\mathbb{P}(\beta>0)>0$.}}

\medskip
Since all mass of $F$ is on $[0,\infty)$,
$G(0)=\int_{-\infty}^0 x^2\,dF(x)=0$.
We identify $\beta^*$ directly.  Since $F$ places no mass on
$(-\infty,0)$ and the integrand $x(x-y)$ vanishes at $x=0$, for $y>0$
\[
    G(y)=\int_{(0,y]} x(x-y)\,dF(x).
\]
The integrand is strictly negative on $(0,y)$ and vanishes at $x=y$, so
$G(y)=0$ if and only if $F$ places no mass on $(0,y)$, i.e.\
$F(y^-)=F(0)$.  (An atom of $F$ at $0$, if present, contributes nothing.)
Let $b:=\inf\{y\ge0:F(y)>F(0)\}$, which is finite since
$\mathbb{P}(\beta > 0) >0$.  If $0<y\le b$, then $F(y^-)=F(0)$ and
$G(y)=0$.  If $y>b$, there is $y'\in(b,y)$ with $F(y')>F(0)$, so
$F(y^-)>F(0)$ and $G(y)<0$.  Together with $G(0)=0$, this shows $b  = \beta^*$ and
\[
    \beta^*=\inf\{y\ge 0:F(y)>F(0)\},
\]
which is the left endpoint of the support of $F$ restricted to
$(0,\infty)$ (with $\beta^*=0$ if that support accumulates at $0$).  In
particular, $F$ places no mass on $(0,\beta^*)$, so
$F(\beta^{*-})=F(0)$ and
\[
    F(\beta^*)=\mathbb{P}\bigl(\beta\in\{0,\beta^*\}\bigr).
\]

Note that $\E[\beta]>0$: if $\E[\beta]=0$ with all mass on $[0,\infty)$
then $\beta=0$ a.s., contradicting $\mathbb{P}(\beta>0)>0$.  We fix an
outcome in $\Omega_0$ at which, in addition, no $\beta_j$ lies in
$(-\infty,0)\cup(0,\beta^*)$; since $F$ places no mass on this set, this
is a probability-one event.  By~\ref{om:s},
$(1/p)\sum_i\beta_i/\delta_i^2>0$ for all large $p$.

\emph{The zero $\beta^*(p)$ of $G_p$.}  All $\beta_j\ge0$, and since
$\F_p(0)\to F(0)<1$ by~\ref{om:F}, some $\beta_j>0$ for all large $p$.  Let
$\beta_{(+)}:=\min\{\beta_j:\beta_j>0\}$.  For $0\le y<\beta_{(+)}$ only
zero betas satisfy $\beta_j\le y$, and they contribute nothing, so
$G_p(y)=G_p(0)=\nu^2/(p\sigma^2)>0$.  The function $G_p$ is continuous
and piecewise affine, with slope
$G_p'(y)=-(\nu^2/p)\sum_{j:\,\beta_j\le y}\beta_j/\delta_j^2$ off the set
$\{\beta_j\}$; this slope is strictly negative for $y>\beta_{(+)}$, and
equals the negative constant $-(\nu^2/p)\sum_j\beta_j/\delta_j^2$ for
$y\ge\max_j\beta_j$.  Hence $G_p$ is positive on $[0,\beta_{(+)}]$,
strictly decreasing on $[\beta_{(+)},\infty)$, and tends to $-\infty$,
so it has a unique zero $\beta^*(p)\in(\beta_{(+)},\infty)$.  In
particular Lemma~\ref{lem:activeset} applies and
$k_p/p=\F_p(\beta^*(p)^-)$.

\emph{Key observation.}  By the choice of outcome, every draw with
$\beta_j\le\beta^*$ satisfies $\beta_j\in\{0,\beta^*\}$, and so
contributes $\beta_j(\beta_j-\beta^*)/\delta_j^2=0$ to $G_p(\beta^*)$.
Hence
\[
    G_p(\beta^*) = \frac{\nu^2}{p\sigma^2} > 0
    \quad\text{for all }p.
\]
Unlike Case~1, where the random walk $S_p$ of Step~5 takes both signs
infinitely often, here no draw contributes to $G_p(\beta^*)$ beyond the
positive constant term.
Since $G_p$ is non-increasing on $[0,\infty)$ with unique zero
$\beta^*(p)$, it follows that $\beta^*(p)>\beta^*$ for all large $p$.

\emph{Convergence $\beta^*(p)\to\beta^*$.}  Fix $\epsilon>0$ and choose
a rational $q\in(\beta^*,\beta^*+\epsilon)$.  Then $G(q)<0$, so
$G_p(q)<0$ for all large $p$ by~\ref{om:G}, and hence
$\beta^*(p)<q<\beta^*+\epsilon$ for all large $p$, since $G_p$ is
non-increasing on $[0,\infty)$ with $G_p(\beta^*(p))=0$.  Combined with
$\beta^*(p)>\beta^*$, this gives $\beta^*(p)\to\beta^*$.

\emph{Upper bound.}  Step~3 of Case~1 uses only Lemma~\ref{lem:activeset},
$\beta^*(p)\to\beta^*$, item~\ref{om:F}, and right-continuity of $F$, all
of which are available here.  Hence
$\limsup_{p\to\infty}k_p/p \le F(\beta^*)$.

\emph{Lower bound.}  Since $\beta^*(p)>\beta^*$ for all large $p$, every
asset with $\beta_i\le\beta^*$ is active, so
\[
    \frac{k_p}{p} = \F_p(\beta^*(p)^-)
    \ge \F_p(\beta^*) \to F(\beta^*)
\]
by~\ref{om:F}.  Combined with the upper bound,
$k_p/p\to F(\beta^*)=\mathbb{P}(\beta\in\{0,\beta^*\})$.
\end{proof}

\subsection{Proof of Theorem \ref{thm:convergence}} \label{sec:convergence}

\begin{proof}
Fix $F\in\mathcal F(\mu,\kappa,C,M)$, let $X\sim F$, and write $\mu_F:=\E_F[X]\ge\mu$.

\emph{$\beta^*$ is well defined.}  The set $\{x\ge0:G(x)=0\}$ is non-empty:
$G(0)=\E_F[X^2;\,X\le0]\ge0$, $G$ is continuous on $[0,\infty)$
(Preliminary~1 in the proof of Theorem~\ref{thm:kp}), and
$G(y)=\E_F[X^2;\,X\le y]-y\,\E_F[X;\,X\le y]\le C-y\,\E_F[X;\,X\le y]\to-\infty$
as $y\to\infty$, since $\E_F[X;\,X\le y]\to\mu_F>0$.  By the intermediate
value theorem $G$ has a zero in $[0,\infty)$, and Lemma~\ref{lem:ybound}
below shows every such zero is at most $C/\mu$, so $\beta^*<\infty$.  By
continuity, $G(\beta^*)=0$.

Throughout, write $\varepsilon:=F(0)$ and
\[
  q \;:=\; F(\beta^*)-F(0),
\]
for the mass $F$ places on $(0,\beta^*]$, so that
$F(\beta^*)=\varepsilon+q$.  The goal is to show
$q\le\Theta^{1/3}\varepsilon^{1/3}$. For convenience let
$D:=C/\mu$.

\begin{lemma}[Upper bound on zeros of $G$]\label{lem:ybound}
If $y^*\ge0$ and $G(y^*)=0$, then $y^*\le D$.
\end{lemma}

\begin{proof}
When $y^*=0$ the bound is trivial.  When $y^*>0$ the condition
$G(y^*)=0$ is
\[
  \int_{-\infty}^{y^*}x(x-y^*)\,dF(x)=0,
\]
or, equivalently, $\E_F[X^2;\,X\le y^*]-y^*\,\E_F[X;\,X\le y^*] =0$.
Adding $\E_F[X(X-y^*);\,X>y^*]$ to both sides gives
\[
  \E_F[X^2]-y^*\mu_F=\E_F[X(X-y^*);\,X>y^*]\ge 0,
\]
since the integrand is non-negative on $\{X>y^*\ge 0\}$.  Hence
$y^*\le\E_F[X^2]/\mu_F\le C/\mu=D$.
\end{proof}

\begin{lemma}[Exact identity at $\beta^*$]\label{lem:identity}
When $\beta^*>0$,
\begin{equation}
  \underbrace{\int_0^{\beta^*}x(\beta^*-x)\,dF(x)}_{=:\,L}
  \;=\;
  \underbrace{%
    \E_F[X^2;\,X\le 0]+\beta^*\,\E_F[|X|;\,X\le 0]
  }_{=:\,R},
  \label{eq:star}
\end{equation}
so that $L=R\ge 0$.
\end{lemma}

\begin{proof}
Split $G(\beta^*)=0$ at $x=0$:
\[
  \int_{-\infty}^0 x(x-\beta^*)\,dF(x)
  +\int_0^{\beta^*}x(x-\beta^*)\,dF(x)=0.
\]
For $x\le 0$ and $\beta^*>0$ we have $x(x-\beta^*)=x^2+\beta^*|x|\ge 0$, so
the first integral equals $R\ge 0$.  Since $x-\beta^*<0$ on $(0,\beta^*)$
the second integral equals $-L\le 0$.  Hence $R-L=0$.
\end{proof}

\begin{lemma}[Upper bound on $R$]\label{lem:Rbound}
$R\le(\kappa+D\sqrt{\kappa})\,\varepsilon$.
\end{lemma}

\begin{proof}
By hypothesis (ii), $\E_F[X^2;\,X\le 0]\le \kappa\varepsilon$.
By the Cauchy--Schwarz inequality,
\[
  \E_F[|X|;\,X\le 0]
  \le \E_F[X^2;\,X\le0]^{1/2}\,F(0)^{1/2}
  \le \sqrt{\kappa}\,\varepsilon.
\]
Using $\beta^*\le D$ from Lemma~\ref{lem:ybound}:
$R\le \kappa\varepsilon+D\sqrt{\kappa}\,\varepsilon=(\kappa+D\sqrt{\kappa})\varepsilon$.
\end{proof}

\begin{lemma}[Lower bound on $L$]\label{lem:lower}
$\displaystyle L\ge\frac{q^3}{27\,M^2}$.
\end{lemma}

\begin{proof}
When  $q=0$ the inequality is trivial.
Assume $q>0$ and set $c:=q/(3M)$.

By the concentration function bound with $x=0$ and $t=\beta^*$,
\[
  q = F(\beta^*)-F(0) \;\le\; M\cdot \beta^*,
\]
so $\beta^*\ge q/M=3c$, and the interval $[c,\beta^*-c]$ is non-empty with
length at least $c>0$.  Applying the concentration function bound to the
two endpoint regions gives
\[
  F(c)-F(0) \;\le\; Mc = \frac{q}{3},
  \qquad
  F(\beta^*)-F(\beta^*-c) \;\le\; Mc = \frac{q}{3},
\]
so the mass on the middle interval satisfies
\[
  F(\beta^*-c)-F(c)
  \;\ge\; q - \frac{q}{3} - \frac{q}{3}
  \;=\; \frac{q}{3}.
\]
On $[c,\beta^*-c]$ we have $x\ge c$ and $\beta^*-x\ge c$, so
$x(\beta^*-x)\ge c^2=q^2/(9M^2)$.  Therefore
\[
  L=\int_0^{\beta^*}x(\beta^*-x)\,dF(x)
  \;\ge\;\int_{[c,\,\beta^*-c]}x(\beta^*-x)\,dF(x)
  \;\ge\;\frac{q^2}{9M^2}\cdot\frac{q}{3}
  =\frac{q^3}{27\,M^2}.
\]
\end{proof}

\medskip
\noindent\textbf{Completion of the proof of Theorem~\ref{thm:convergence}.}

\smallskip
\textit{Case $\beta^*=0$.}  Then $q=0$ and $F(\beta^*)=F(0)=\varepsilon$.

\smallskip
\textit{Case $\beta^*>0$.}
Combining the identity $L=R$ (Lemma~\ref{lem:identity}) with the
bounds of Lemmas~\ref{lem:Rbound} and~\ref{lem:lower}:
\[
  \frac{q^3}{27\,M^2}\;\le\;L=R\;\le\;(\kappa+D\sqrt{\kappa})\varepsilon.
\]
Solving for $q$:
\[
  q\;\le\;\left({27\,(\kappa+D\sqrt{\kappa})\,M^2}\right)^{\!1/3}
  \varepsilon^{1/3}
  \;=\;\Theta^{1/3}\varepsilon^{1/3}.
\]
In both cases
$F(\beta^*)=\varepsilon+q\le F(0)+\Theta^{1/3}F(0)^{1/3}$.
The final assertion of the theorem follows since the right-hand side
depends on $F$ only through $F(0)$ and is increasing in $F(0)$.
\end{proof}

\section{Acknowledgments}

The authors thank Lisa Goldberg, Nick Gunther, Alex Shkolnik, and Hubeyb Gurdogan for helpful conversations. 
All errors are our own. 

\section{Appendix: Use of AI Assistance}

During the preparation of this paper the authors used
Claude (Anthropic) large language models (hereafter: Claude) as interactive research assistants. During March 2026 for initial preparation, these were Sonnet 4.6 and Opus 4.6; during the revision period in October-November 2026 we used Opus 5.5. 

All results and proofs were initiated and verified by the authors, who take full responsibility for the content. Claude is not an author. 

The statement and proof of the foundational Theorem~\ref{thm:1} were developed in summer and fall of 2025 prior to our access to LLMs.  The subsequent asymptotic theorems were developed in initial versions by the authors, who then used Claude to help explore proof ideas, generalizations, and weaker hypotheses, to locate useful existing theory, and to find related papers in the literature. A practice of query, evaluation, correction, and redirection became a fruitful back-and-forth that accelerated the research process. The authors could likely have done everything without Claude at the cost of a much longer development time. 

 Under the authors' guidance through interactive sessions, Claude drafted the Python code  that generated the figures and tables, and helped find the non-convergence example of Section~\ref{subsec:example}.
 Quickly producing LaTeX code is also one of Claude's strengths.  At times we asked for text drafts or outlines that we took as expositional suggestions, often requiring correction or editing, frequently rejected, occasionally erroneous, but usually helpful.  The authors carefully vetted all output from Claude.

 For the main mathematical content, Claude helped us discover more general hypotheses for Theorem~\ref{thm:kp} and an initially overlooked edge case where the limit can fail to exist. It also was useful in finding gaps in proof versions, and suggesting the random walks recurrence argument.  Claude was also  very helpful in developing a sequence of more efficient proofs of Theorem~\ref{thm:convergence}, leading to the existing version.

 The literature summary was accelerated by initial Claude reviews; every paper was read by the authors, who wrote the final drafts.

\bibliography{ref-master} 

@book{durrett2019,
  author    = {Durrett, Rick},
  title     = {Probability: Theory and Examples},
  edition   = {5},
  series    = {Cambridge Series in Statistical and Probabilistic Mathematics},
  volume    = {49},
  publisher = {Cambridge University Press},
  address   = {Cambridge},
  year      = {2019},
  note      = {{Version of January 11, 2019, available at the author's website}}
}

@article{traut2023,
	author = {Traut, J.},
	journal = {Financ. Mark. Portf. Manag.},
	number = {3},
	pages = {297-324},
	title = {What we know about the low-risk anomaly: a literature review},
	volume = {37},
	year = {2023}}

@article{steland2024,
	author = {Steland, A.},
	journal = {Math. Financ. Econ.},
	pages = {151-170},
	title = {Are minimum variance portfolios in multi-factor models long in low-beta assets?},
	volume = {18},
	year = {2024}}

@unpublished{gks2025,
	author = {N. L. Gunther and A. Kercheval and O. Sowunmi},
	note = {arxiv.org/abs/2603.07692},
	title = {Understanding the Long-Only Minimum Variance Portfolio},
	year = {2026}}

@book{rockafellar1970,
	author = {R. T. Rockafellar},
	publisher = {Princeton University Press},
	title = {Convex Analysis},
	year = {1970}}

@article{green-holliffield1992,
	author = {R. C. Green and B. Hollifield},
	journal = {The Journal of Finance},
	month = {December},
	number = {5},
	pages = {1785-1809},
	title = {When will mean-variance efficient portfolios be well diversified?},
	volume = {47},
	year = {1992}}

@article{jagann2003,
	author = {R. Jagannathan and T. Ma},
	journal = {The Journal of Finance},
	month = {August},
	number = {4},
	pages = {1651-1683},
	title = {Risk reduction in large portfolios: why imposing the wrong constraint helps},
	volume = {LVIII},
	year = {2003}}

@article{markowitz1952,
	author = {H. Markowitz},
	journal = {Journal of Finance},
	number = {1},
	pages = {77-92},
	title = {Portfolio Selection},
	volume = {7},
	year = {1952}}

@article{sharpe1963,
	author = {William F. Sharpe},
	journal = {Management Science},
	month = {Jan.},
	number = {2},
	pages = {277-293},
	title = {A simplified model for portfolio analysis},
	volume = {9},
	year = {1963}}

@article{best1992,
	author = {M. J. Best and R. R. Grauer},
	journal = {Journal of Financial and Quantitative Analysis},
	number = {4},
	pages = {513-537},
	title = {Positively weighted minimum-variance portfolios and the structure of asset expected returns},
	volume = {27},
	year = {1992}}

@book{boyd-vandenberghe2004,
	address = {Cambridge, UK},
	author = {Boyd, S. and Vandenberghe, L.},
	publisher = {Cambridge Univ. Press},
	title = {Convex Optimization},
	year = {2004}}

@article{cst2011,
	author = {R. Clarke and H. De Silva and S. Thorley},
	journal = {The Journal of Portfolio Management},
	pages = {31--45},
	title = {Minimum-Variance Portfolio Composition},
	volume = {Winter},
	year = {2011}}

@article{qi2021,
	author = {H. Qi},
	journal = {Operations Research Letters},
	pages = {795-801},
	title = {On the long-only minimum variance portfolio under single factor model},
	volume = {49},
	year = {2021}}

@book{beck2023,
	author = {Amir Beck},
	edition = {2nd},
	publisher = {SIAM},
	title = {Introduction to Nonlinear Optimization: Theory, Algorithms, and Applications with Python and MATLAB},
	year = {2023}}

@unpublished{bernstein2025,
	author = {A. Bernstein and A. Shkolnik},
	month = {September},
	note = {preprint},
	title = {Asymptotics of Quadratic Forms on a Simplex},
	year = {2025}}

@article{chung-fuchs1951,
	author = {K. L. Chung and W. H. J. Fuchs},
	journal = {Memoirs of the American Mathematical Society},
	pages = {1--12},
	title = {On the distribution of values of sums of random variables},
	volume = {6},
	year = {1951}}

\end{document}